\documentclass[12pt]{article}
\usepackage{times,amssymb,latexsym,amsmath,setspace}
\usepackage[a4paper,top=1in,bottom=1in,left=0.8in,right=0.8in]{geometry}
\usepackage[auth-sc,affil-it]{authblk}
\usepackage{array,url}
\usepackage{enumitem}
\usepackage{multirow,tabularx}
\usepackage[comma]{natbib}
\usepackage{bibunits}
\defaultbibliographystyle{apalike}
\defaultbibliography{causal}
\usepackage{bm} 
\usepackage{amsthm}
\usepackage{makecell}
\usepackage{xcolor}

\usepackage{booktabs}

\usepackage{caption} 

\usepackage{graphicx}
\usepackage[figuresright]{rotating}
\newcommand{\ind}{\rotatebox[origin=c]{90}{$\models$}}

\newcolumntype{M}[1]{>{\raggedright\arraybackslash}m{#1}}
\newcolumntype{C}[1]{>{\centering\arraybackslash}m{#1}}
\newcommand{\estcell}[1]{\begin{minipage}[c]{\linewidth}\raggedright\footnotesize #1\end{minipage}}

\usepackage{tikz}
\usetikzlibrary{arrows.meta, positioning}

\newtheoremstyle{note}{2pt}{2pt}{}{}{\bfseries}{:}{.5em}{}

\theoremstyle{note}
\newtheorem{theorem}{Theorem}
\newtheorem{lemma}{Lemma}

\newtheorem{remark}{Remark}
\newtheorem{example}{Example}
\newtheorem{assumption}{Assumption}

\newtheorem{prop}{Proposition}
\allowdisplaybreaks

\newtheorem*{assumption*}{Assumption}

\newcounter{suppsection}
\renewcommand{\thesuppsection}{S\arabic{suppsection}}
\newcommand{\suppsection}[2]{\refstepcounter{suppsection}\section*{\thesuppsection: #1}\label{#2}\renewcommand{\thetable}{\thesuppsection.\arabic{table}}\renewcommand{\thefigure}{\thesuppsection.\arabic{figure}}\setcounter{table}{0}\setcounter{figure}{0}}

\def\bSig\mathbf{\Sigma}

\setlist[itemize]{topsep=2pt,itemsep=1pt,parsep=0pt,partopsep=0pt}

\begin{document}
\begin{bibunit}
\title{Causal Inference for All: Marginal Estimands for Outcomes Truncated by Death}
\author[1]{Ruixuan Zhao}
\author[2]{Mats Stensrud}
\author[1,3]{Linbo Wang}
\affil[1]{Department of Computer and Mathematical Sciences, University of Toronto Scarborough}
\affil[2]{Institute of Mathematics, \'Ecole Polytechnique F\'ed\'erale de Lausanne}
\affil[3]{Department of Statistical Sciences, University of Toronto}
\renewcommand\Authfont{\normalsize}
\renewcommand\Affilfont{\small}
\date{}
  \maketitle

\begin{abstract}
In longitudinal studies, outcomes of interest are often truncated by death, meaning that they are only observed or well-defined conditional on intercurrent events such as survival. Existing strategies face a trade-off: causally interpretable estimands, such as survivor average causal effects, target a latent subgroup, whereas while-alive and composite summaries apply to the full population but are difficult to interpret as causal effects on the non-mortality outcome. We address these challenges by introducing methodology for a new set of estimands that (i) concern the entire population, (ii) remain causally interpretable, and (iii) leverage the longitudinal data commonly available in studies with outcomes truncated by death. The set of estimands includes single-world marginal separable effects that generalize conditional separable effects to full-population summaries. We develop identification and estimation results for these estimands and apply the methodology in a reanalysis of a prostate cancer trial, highlighting how different estimands can yield different treatment conclusions.
\end{abstract}


\section{Introduction}

In longitudinal studies, outcomes are often only well-defined conditional on a post-treatment event, such as survival.  This creates a fundamental challenge for causal inference because contrasts restricted to those who survive need not have a causal interpretation. Truncation by death is the canonical example: in studies of quality of life (QoL), outcomes are well-defined only while participants are alive. In HIV vaccine studies, virologic and immunologic outcomes such as HIV viral load and CD4 cell count are often of interest only among individuals who become infected \citep{gilbert2003sensitivity, wang2017causal}. 
In labor economics, the effect of a job training program on salary is of interest only for individuals who obtain employment, with job attainment truncating the earnings outcome.

In studies with a single follow-up time, a naive yet common practice is to restrict comparisons to observed survivors in each treatment arm. In longitudinal settings, the while-alive strategy extends this survivor-restricted idea by aggregating each participant's outcomes over observed survival time.
Because the strategy includes outcomes only over observed survival time, it can be viewed as a longitudinal analogue of survivor-restricted comparisons, even though the International Council for Harmonisation (ICH) E9(R1) Addendum \citep{ich2019addendum} does not formally define such an estimand.
The comparisons are not causally interpretable because they contrast outcomes from different subpopulations or from different time periods under different treatment arms, thereby introducing selection bias. Another common approach is to treat truncation by death as a missing-data problem, either by estimating outcomes under a hypothetical intervention on the intercurrent event, that is, under a counterfactual scenario in which the intercurrent event does not occur \citep[the hypothetical strategy;][]{ich2019addendum}, or by assigning a prespecified value or rank to outcomes after death \citep[the composite strategy;][]{ich2019addendum}. However, choosing an appropriate prespecified value or rank can be controversial and of limited interpretability \citep{lachin2020worst}. More fundamentally, outcomes such as QoL are arguably undefined, rather than missing, for individuals who die \citep{little2012prevention}.

To address these challenges, a strategy widely advocated in the causal inference literature is to focus on the always-survivor group, that is, the individuals who would survive under either treatment \citep[][\S 12.2]{robins1986new}. The resulting survivor average causal effect (SACE) \citep{rubin2006causal} has generated a large literature on identification, estimation, and sensitivity analysis \citep[e.g.][]{gilbert2003sensitivity,ding2011identifiability,wang2017causal,wang2017identification}, and is also included among the estimands discussed in the ICH E9(R1) Addendum \citep{ich2019addendum}. Yet the SACE has had limited practical uptake, partly because it targets an unobserved and often healthier principal stratum rather than the full population relevant for clinical or policy decisions \citep{bornkamp2021principal,stensrud2022translating}. This limitation is not merely technical: because principal-stratum membership is unknown, recommendations based on the SACE may be difficult to apply to individual patients and may expose individuals outside the principal stratum to an ineffective or even harmful treatment, raising ethical concerns \citep{robins1986new,hernan2018cautions,stensrud2022separable}.

Recently, \cite{stensrud2023conditional} proposed the conditional separable effect (CSE), which conceptualizes treatment as two components operating through distinct pathways on survival and the outcome. The CSE avoids the cross-world conditioning in the SACE, but still conditions on survival under an intervention on the survival-related component. It is therefore causally interpretable but remains a survival-conditioned estimand, rather than a full-population summary of longitudinal outcomes.

\begin{figure}[!htbp]
    \centering
\includegraphics[width=0.85\linewidth]{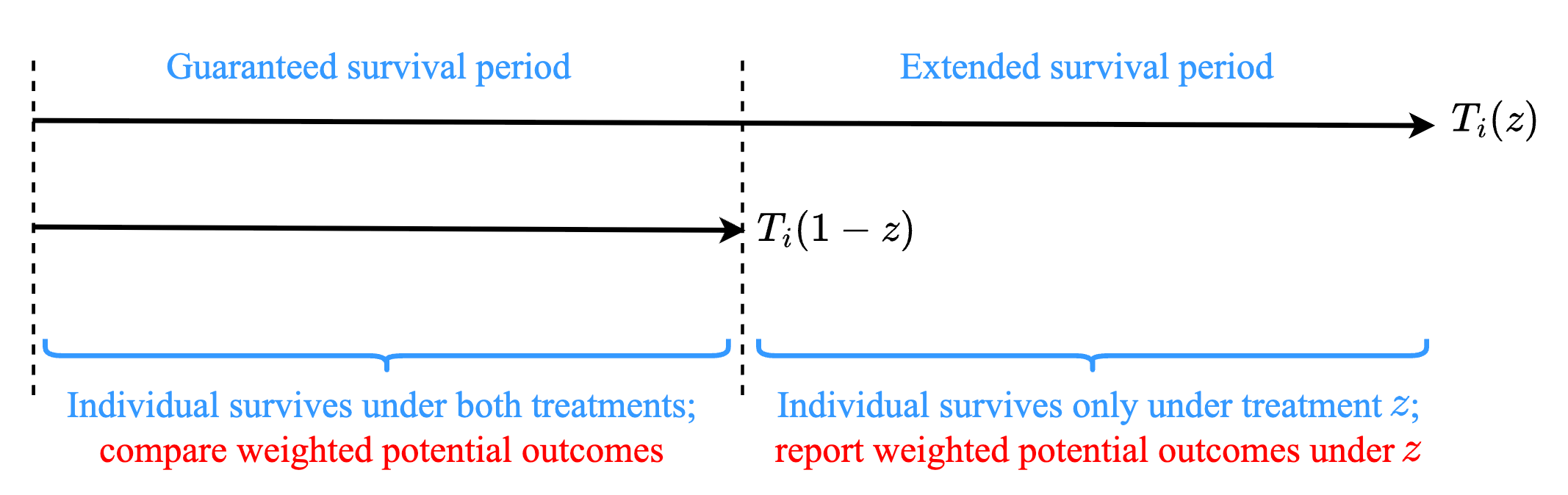}
    \caption{Illustration of the ``while guaranteed-survival'' and ``while extended-survival'' estimands, each associated with a distinct time period, when $T_i(z)\geq T_i(1-z)$ for individual $i$.
    }
    \label{fig:survival_timeline}
\end{figure}

To reconcile the discrepancy between causally interpretable estimands and the summaries commonly used in practice, we separate two questions that are often conflated. First, over the period during which the outcome would be well-defined under either treatment, what is the treatment effect on the non-mortality outcome? Second, when one treatment prolongs survival, what is the outcome experience during the additional survival time? Figure~\ref{fig:survival_timeline} illustrates this partition into guaranteed survival and extended survival periods. This separation yields full-population causal estimands while making explicit the cross-world quantities required to define and identify them.

Our first contribution is a class of ``while guaranteed-survival'' estimands. Let \(T(z)\) denote the potential follow-up time under treatment \(z\). These estimands summarize treatment effects over the guaranteed survival period \([0,T(0)\land T(1)]\), where the outcome would be well-defined regardless of treatment. They retain the full-population appeal of while-alive and composite summaries, while avoiding treatment-specific time domains and therefore defining a coherent causal contrast. By allowing exit-time, average, cumulative, or AUC-based weighting schemes, they can be tailored to different clinical questions about the outcome experience before death truncates comparability.

Second, we introduce ``while extended-survival'' estimands for the period during which an individual would survive under one treatment but not under the other. When \(T(z)> T(1-z)\), this period is the interval \((T(1-z),T(z)]\) shown in Figure~\ref{fig:survival_timeline}. This part of the trajectory is ignored by the SACE and guaranteed-survival contrasts, but is often implicitly captured by standard clinical summaries such as while-alive or composite strategies. For QoL outcomes, explicitly targeting this period distinguishes treatments that prolong survival with meaningful QoL from those that prolong survival in a severely compromised state, sometimes described as a ``living death.''

Third, we develop single-world marginal separable effects that generalize recent work on separable effects \citep{stensrud2022separable,stensrud2023conditional}. Under a conceptual decomposition of treatment into an outcome-related component \(Z_Y\) and a survival-related component \(Z_S\), these estimands compare interventions that vary \(Z_Y\) while holding \(Z_S\) fixed. They define population-level contrasts without relying on the cross-world survival times used by the while guaranteed-survival and while extended-survival estimands. This yields a complementary single-world framework, and we further clarify conditions under which the while guaranteed-survival and while extended-survival estimands correspond to marginal separable effects.

Beyond defining these estimands, we study how they can be learned from observed longitudinal trial data. For the while guaranteed-survival and while extended-survival estimands, we derive identification results using a substitution-variable approach that links the required cross-world quantities to observed longitudinal data, and we propose compatible plug-in estimators. For marginal separable effects, we establish identification under a treatment decomposition and construct estimators based on the resulting identification formulas. Throughout, these assumptions are used for identification and estimation rather than for defining the estimands themselves.

The rest of the paper is organized as follows. Section 2 introduces the longitudinal setup and reviews existing strategies for outcomes truncated by death. Section 3 defines the while guaranteed-survival and while extended-survival estimands. Sections 4 and 5 develop identification and estimation results for these estimands and for marginal separable effects, respectively. Section 6 reanalyzes a prostate cancer trial, and Section 7 discusses implications and limitations.

\section{Preliminaries}

\subsection{Basic setup}\label{sec:setup}

Consider a longitudinal study, where individuals are assigned a binary treatment $Z\in \{0,1\}$ at baseline and baseline covariates $L^0$ are observed. By convention, we set $Z=1$ as active treatment and $Z=0$ as control. We consider discrete follow-up time points indexed by
$t=1,\ldots,t_{\max}$, where $t_{\max}$ denotes the index of the final scheduled follow-up visit. Let $S^t$ denote the survival status at time point $t$, with $S^t=1$ if the individual was alive and $S^t=0$ otherwise. Similarly, we denote $Y^t$ as the non-mortality outcome at time $t$, where $Y^t$ is well-defined only if $S^t=1$. We assume $S^0=1$, and the baseline outcome $Y^0$ is well-defined for each individual, with $Y^0\in L^0$. Figure \ref{fig:longtbd} shows a diagram to illustrate the data structure introduced above. In addition, we define the observed follow-up time as $T=\max\{t\in\{0,1,...,t_{\max}\}: S^t=1\}$.
We use overbars to denote the history of a random variable over time $t$, such as $\bar{S}^t=(S^1,..., S^t)$, and underbars to denote its future, such as $\underline{S}^t = (S^t,..., S^{t_{\max}})$.

\begin{figure}[!htbp]
    \centering
\begin{tikzpicture}[node distance=2cm, >={Stealth[round]}]

    \node (A) at (0,0) {};
    \node (Z) [right=of A] {$Z$};
    \node (Y1) [right=of Z] {$Y^1$};
    \node (Y2) [right=of Y1] {$Y^2$};
    \node (Y3) [right=of Y2] {$Y^3$};
    \node[above=of Z, yshift=-2.2cm] {\begin{tabular}{c} Binary \\ treatment \end{tabular}};
    \node[above=of Y1, yshift=-2.2cm] {\begin{tabular}{c} Outcome \\
    at time 1 \end{tabular}};
     \node[above=of Y2, yshift=-2.2cm] {\begin{tabular}{c} Outcome \\
    at time 2 \end{tabular}};
    \node[above=of Y3, yshift=-2.2cm] {\begin{tabular}{c} Outcome \\
    at time 3 \end{tabular}};
    \node (S1) [below=of Y1] {$S^1$};
    \node (S2) [right=of S1] {$S^2$};
    \node (S3) [right=of S2] {$S^3$};
    \node[below=of S1, yshift=2.2cm] {\begin{tabular}{c} Survival status\\ at time 1 \end{tabular}};
    \node[below=of S2, yshift=2.2cm] {\begin{tabular}{c} Survival status\\ at time 2 \end{tabular}};
    \node[below=of S3, yshift=2.2cm] {\begin{tabular}{c} Survival status\\ at time 3 \end{tabular}};

    \draw[->] (Z) -- (Y1);
    \draw[->] (Y1) -- (Y2);
    \draw[->] (Y2) -- (Y3);

    \draw[->] (Z) -- (S1);
    \draw[->] (S1) -- (S2);
    \draw[->] (S1) -- (Y1);
    \draw[->] (S2) -- (Y2);
    \draw[->] (S2) -- (S3);
    \draw[->] (S3) -- (Y3);

\end{tikzpicture}
\caption{Longitudinal data structure with outcomes truncated by death. Survival status \(S^t\) determines whether the outcome \(Y^t\) is well-defined at each follow-up time.}
\label{fig:longtbd}
\end{figure}
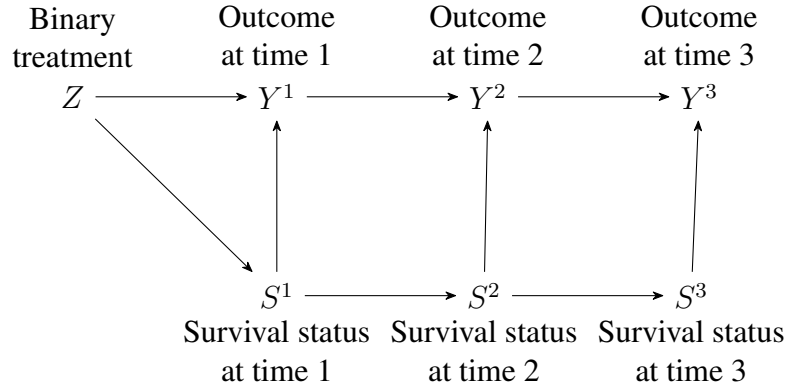

In the potential outcomes framework, for $z=0,1$, let the potential survival status at time $t$ be $S^t(z)$, which is the survival status that would be observed were the individual assigned treatment $z$. Similarly, we denote $Y^t(z)$ and $T(z)$ as the potential outcome at time $t$ and potential follow-up time under treatment $z$, respectively. Also, note that $Y^t(z)$ is well-defined only if $S^t(z)=1$. For $z=0,1$, let $T(z)\land T(1-z)$ and $T(z)\lor T(1-z)$ denote, respectively, the minimal and maximal follow-up time across two treatments. Under the standard consistency condition, the observed survival statuses, the survival time, and the outcomes satisfy $S^t = ZS^t(1)+ (1-Z)S^t(0)$, $T=ZT(1)+(1-Z)T(0)$, and $Y^t = ZY^t(1)+ (1-Z)Y^t(0)$, respectively. Furthermore, the observed samples $\{L^0_i,Z_i, \bar{S}_i^{t_{\max}}, \bar{Y}_i^{t_{\max}}\}_{i=1}^n$ are independently drawn from the joint distribution of $\{L^0, Z,\bar{S}^{t_{\max}}, \bar{Y}^{t_{\max}}\}$. Throughout, for a random vector $V$, let $\mathcal{V}$ denote its support, and use lowercase letters to denote its realizations. For two random vectors $A$ and $B$, denote $p_{A|B}(a|b)$ as the conditional density or conditional mass function of $A=a$ given $B=b$.

\subsection{Survivor average causal effect  }

The survivor average causal effect (SACE) is a widely studied principal stratum estimand \citep{robins1986new,rubin2006causal}. Specifically, we use $G^t$ to denote the survival type at time $t$, for $t=1,...,t_{\max}$, which is defined as follows: $G^t=LL$ if $S^t(1)=S^t(0)=1$, $G^t=LD$ if $S^t(1)=1$ and $S^t(0)=0$, $G^t=DL$ if $S^t(1)=0$ and $S^t(0)=1$, and $G^t=DD$ if $S^t(1)=S^t(0)=0$. Then, the SACE at time $t$ is 
$\mathbb{E}[Y^t(1)-Y^t(0)|G^t=LL]$, or equivalently $\mathbb{E}[Y^t(1)-Y^t(0)|T(1)\geq t, T(0)\geq t]$. In longitudinal settings, \citet{grossi2023bayesian} evaluate treatment effects among always-survivor groups across a sequence of time points, using estimands of the form $\mathbb{E}[Y^r(1)-Y^r(0)\mid T(1)\ge t, T(0)\ge t]$ for $r \le t$. For fixed $t$, these estimands describe the treatment-effect trajectory over $r=1,\ldots,t$ within the same always-survivor subgroup. However, comparisons across $t$ are inappropriate because they involve different latent subgroups \citep{comment2025survivor}. Accordingly, \citet{grossi2023bayesian} also report posterior membership probabilities for the longitudinal principal strata, clarifying their proportions and baseline characteristics.

\subsection{Separable effect}\label{sec:separable}

\citet{stensrud2023conditional} considered an extension of the setup in Section \ref{sec:setup} by accounting for observed time-varying covariates $L^t$ at times $t=1,\ldots,t_{\max}$, which may be treatment-induced confounders between $S^t$ and $Y^t$. This extended setting is shown in Figure \ref{fig:Longi-SE-2}(a). The separable effects start with the premise that the original treatment \(Z\) in the current trial can be conceptualized as comprising two binary components, \(Z_Y \in \{0,1\}\) and \(Z_S \in \{0,1\}\). The component \(Z_Y\) directly affects the longitudinal outcome process \(\bar{Y}^{t_{\max}}\) without influencing survival \(\bar{S}^{t_{\max}}\), whereas \(Z_S\) directly affects survival and may influence \(\bar{Y}^{t_{\max}}\) only through the covariate process \(\bar{L}^{t_{\max}}\) {and survival \(\bar{S}^{t_{\max}}\)}. More specifically, let  $Y^t(z_Y,z_S)$, $L^t(z_Y,z_S)$, $S^t(z_Y,z_S)$ and $T(z_Y,z_S)$ denote the potential outcome at time $t$, the potential time-varying covariate at time $t$, the potential survival status at time $t$, and the potential follow-up time, respectively, under treatment components $(z_Y,z_S)$.
The following assumptions formalize the assumptions on $Z_Y$ and $Z_S$.

\begin{assumption}[$Z_Y$ partial isolation]\label{ass:Z_yPI}
 $S^t(z_Y,z_S)=S^t(z_S)$
for $z_Y,z_S \in \{0,1\}$ and $t=1,...,t_{\max}$.
\end{assumption}

\begingroup
\renewcommand{\theassumption}{2*}
\begin{assumption}[Weakly $Z_S$ partial isolation]\label{ass:Z_S-weak}
 Let $Y^t(z_Y,\bar{l}^t,\bar{s}^t)$ denote the potential outcome at time $t$ under an intervention that sets $Z_Y=z_Y$, $\bar{L}^t=\bar{l}^t$, and $\bar{S}^t=\bar{s}^t$. Then
\(
Y^t(z_Y,z_S)=Y^t\bigl(z_Y, \bar{L}^t(z_Y,z_S), \bar{S}^t(z_S)\bigr), 
\) for $t=1,\ldots,t_{\max}.$
\end{assumption}
\addtocounter{assumption}{-1}
\endgroup

Assumption \ref{ass:Z_S-weak} can be controversial because it implicitly assumes that it is possible to intervene on $\bar{S}^t$. Instead, a common alternative in the separable effects literature \citep{stensrud2022separable,stensrud2023conditional} is to consider a hypothetical four-arm trial \(\mathcal{C}\) in which \(Z_Y\) and \(Z_S\) are jointly randomized in the same population. Let $V(\mathcal{C})$ denote the random variable $V$ under the treatment assignment in trial $\mathcal{C}$, to distinguish it from the observed $V$ in the current two-arm trial. Since $L^0$ is the baseline covariate measured prior to treatment assignment and is unaffected by treatment, we use $L^0$, rather than $L^0(\mathcal{C})$, under trial $\mathcal{C}$. Then, the following assumption captures the notion that $Z_S$ may influence $Y^t$ only through $(\bar{S}^t,\bar{L}^t)$ for $t=1,\ldots,t_{\max}$.

\begin{assumption}[$Z_S$ dismissible component condition]\label{ass:Z_S-DCC}
 For $t=1,...,t_{\max}$,  $$Z_S(\mathcal{C})\ind Y^t(\mathcal{C})\mid Z_Y(\mathcal{C}), S^t(\mathcal{C})=1, \bar{L}^t(\mathcal{C}), L^0.$$
\end{assumption}

\begin{figure}[!htbp]
    \centering
    \includegraphics[width=0.95\linewidth]{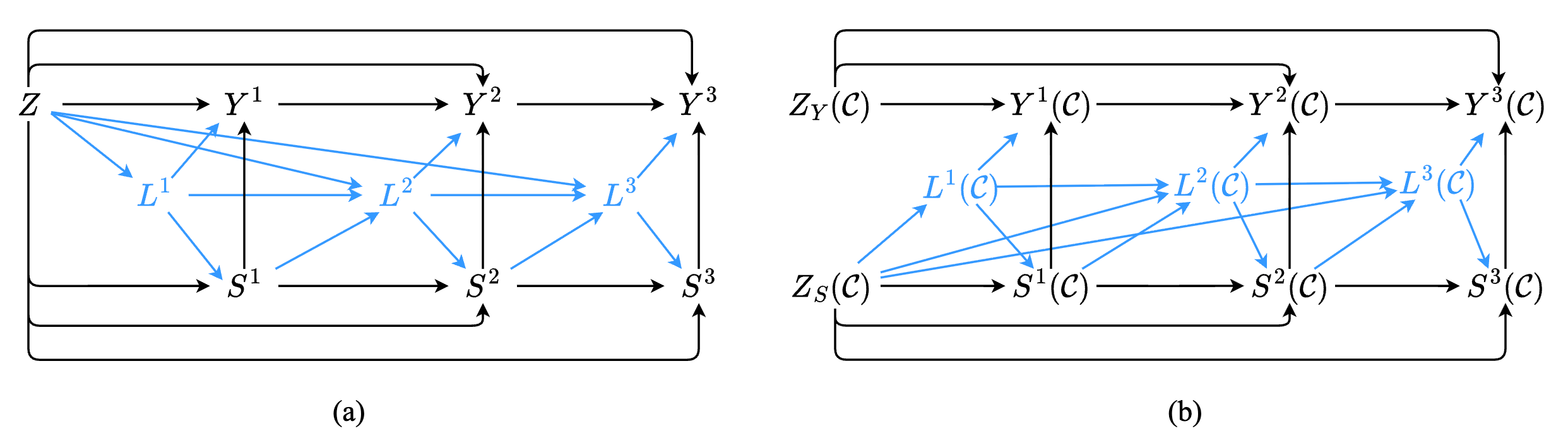}
    \caption{Two causal diagrams illustrating the separable effects framework: panel (a) shows the causal graph for the current study, and panel (b) shows the causal graph under the four-arm trial $\mathcal{C}$. Baseline covariates $L^0$ are omitted in both graphs for brevity. }
    \label{fig:Longi-SE-2}
\end{figure}

Assumption \ref{ass:Z_S-DCC} can be read off from a causal graph; see for example Figure \ref{fig:Longi-SE-2}(b). 
The effect of $Z_Y$ on $Y^t$ can be quantified by the conditional separable effect (CSE), $\mathbb{E}[Y^t(1, z_S) - Y^t(0,z_S)\mid S^t(z_S)=1]$,
which is the average causal effect of treatment $Z_Y$ on $Y^t$ had all individuals been assigned the treatment $Z_S=z_S$, among those who would be alive under $Z_S=z_S$.  
Compared with the SACE, the CSE is single-world because it is defined under a single intervention regime on the treatment components. The conditioning event \(\{S^t(z_S)=1\}\) could, in principle, be identified in the hypothetical four-arm trial, whereas always-survivor membership is generally not identifiable.

Next, we adopt a consistency assumption that links the potential outcomes under the decomposed treatment representation $(Z_Y, Z_S)$ to those under the observed treatment $Z$. 

\begin{assumption}[Intervention consistency]\label{ass:MTA}
For $z' \in \{0,1\}$, $Y^t(z_Y=z', z_S=z') = Y^t(z=z')$, $S^t(z_S=z') = S^t(z=z')$ and $L^t(z_Y=z', z_S=z') = L^t(z=z')$.
\end{assumption}

We compare the CSE evaluated at $z_S=0$ with the SACE under the separable effects framework with Assumptions \ref{ass:Z_yPI}, \ref{ass:Z_S-DCC}, and \ref{ass:MTA}. Recall that at time $t$, the SACE can be expressed as 
\(
\mathbb{E}[Y^t(1,1)-Y^t(0,0)\mid S^t(1)=S^t(0)=1].
\)
Under monotonicity, $S^t(1)\geq S^t(0)$, a.s., both the CSE with $z_S=0$ and the SACE concern the same subpopulation 
$\{i: S_i^t(1)=S_i^t(0)=1\}$ for $t=1,\ldots,t_{\max}$. However, the SACE captures both the direct effect of \(Z_Y\) on \(Y^t\) and the effect of \(Z_S\) mediated through \(\bar{L}^t\), whereas the CSE isolates the effect of \(Z_Y\) by fixing \(z_S=0\). They coincide if there are no causal paths from \(Z_S\) to \(Y^t\) through \(\bar L^t\) \citep{stensrud2023conditional}.

\subsection{Related work}\label{sec:related}

Among existing approaches for longitudinal outcomes truncated by death, the while-alive strategy is closely connected to the estimands considered in this paper because it is also designed to summarize longitudinal outcome trajectories. Building on \citet{wei2023properties,janvin2024causal}, we represent this class of ``while-alive'' estimands as follows:
\begin{align}\label{esimand-while-alive}
\lambda(z) := \mathbb{E}\Big[\sum_{r=0}^{T(z)} \zeta^r_{T(z)}Y^r(z)\Big],
\end{align}
where $\zeta^r_{T(z)}$ denotes a weight that may depend on $T(z)$, the last time point at which the individual would be alive under treatment $z$. Although $\lambda(z)$ is a population-level summary, the contrast $\lambda(1)-\lambda(0)$ compares outcomes accumulated over different time periods across treatment arms and therefore generally conflates treatment effects on the non-mortality outcome with those on survival. As shown by \citet{janvin2024causal}, even when the treatment has no direct effect on the outcome, such contrasts can differ from zero purely due to differences in survival times between arms.

Other estimands for longitudinal outcomes truncated by death, including the recently developed PLOT estimand of \citet{baklicharov2025weakening}, are discussed in Supplementary Material Section~\ref{supple:existing-estimands}, and Table \ref{tab:sum-est} summarizes their mathematical forms and main limitations.

\section{Full-population estimands for outcomes truncated by death}
\label{sec:Estimand}

We now introduce a set of estimands that, as illustrated in Figure~\ref{fig:survival_timeline}, characterizes the potential outcome trajectories of each individual \(i\) across two distinct survival periods: the guaranteed survival period and the extended survival period. The guaranteed survival period, \(0 \leq t \leq T_i(1)\wedge T_i(0)\), is the period during which the individual would be alive under both treatments. During this period, both \(Y_i^t(1)\) and \(Y_i^t(0)\) are well-defined, so treatment arms can be compared using common weights chosen to reflect the clinical or scientific objective. The extended survival period, \(T_i(1)\wedge T_i(0) < t \leq T_i(1)\vee T_i(0)\), is the additional survival time under the treatment that would prolong survival. During this period, only one of the two potential outcome paths, \(Y_i^t(1)\) or \(Y_i^t(0)\), is well-defined, and summaries over this period describe the quality (or burden) of the additional survival time.

\subsection{A class of {while guaranteed-survival} estimands}\label{sec:MargEstimand}

To develop estimands for the full population during the guaranteed survival period, in which both potential outcomes are well-defined, we first illustrate the key ideas using a toy example. Figure \ref{fig:toyEx} gives graphical illustrations for the example with two individuals $i=1,2$ and three discrete follow-up times ($t=1,2,3$) at 3, 6, and 12 months. We denote $Y^t_i(z)$ and $T_i(z)$ as the potential outcome at time $t$ and the potential follow-up time under treatment $z$ for individual $i$. Also, let $Y_i^0$ be the baseline outcome for individual $i$. Then, we use this example to motivate the estimands that align with different research aims.

\begin{figure}[!htbp]
    \centering
    \includegraphics[width=0.6\linewidth]{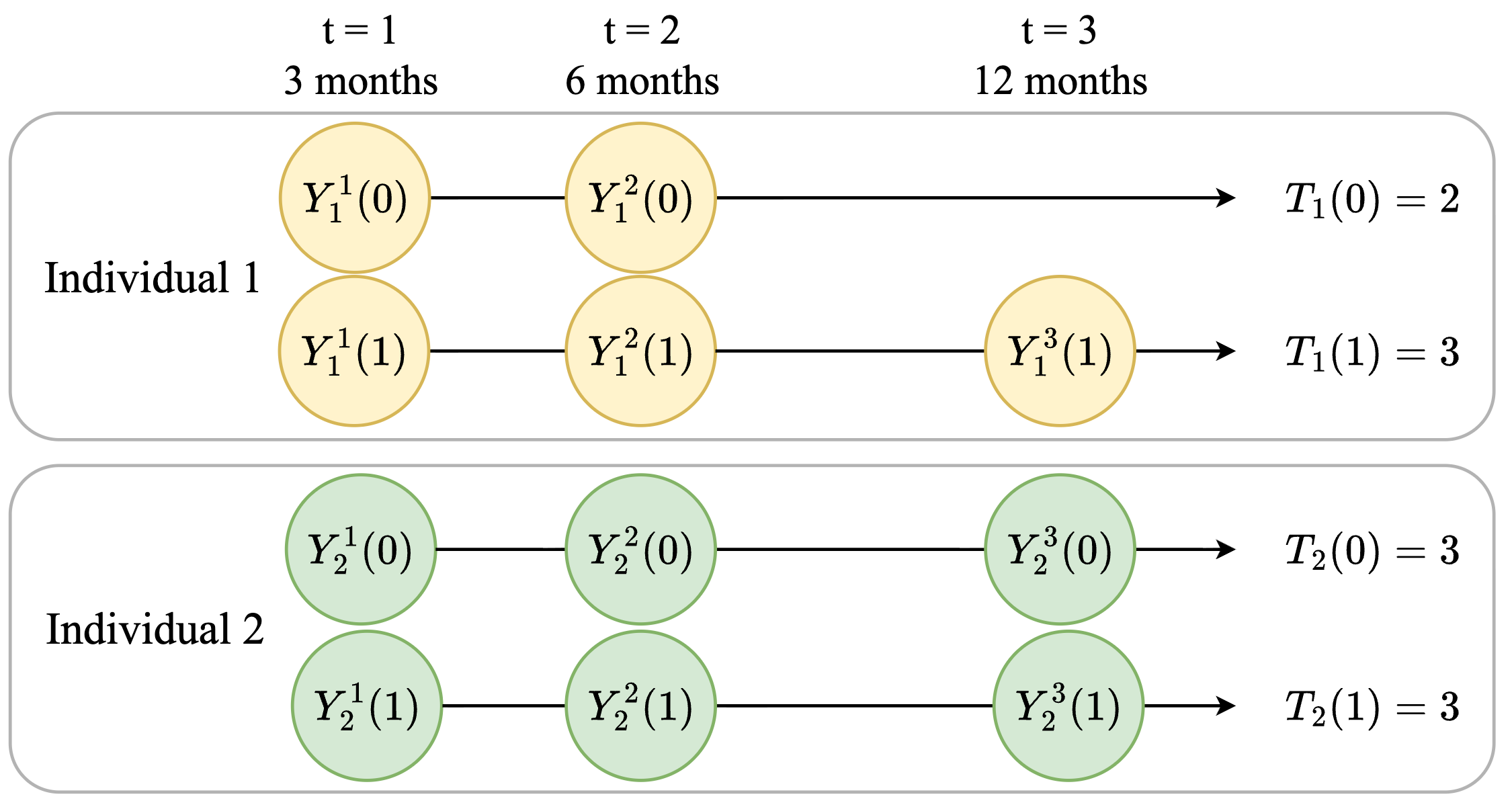}
    \caption{A toy example with two individuals to illustrate our proposed estimands.}
    \label{fig:toyEx}
\end{figure}

\begin{example}[Exit time] \label{exa:exit-w}
 Consider an estimand that summarizes outcomes before death or at the end of the study. In Figure \ref{fig:toyEx}, the estimand takes the form $
\mu(z) = \frac{1}{2}Y_1^2(z) + \frac{1}{2}Y_2^3(z),$
which averages the potential outcomes at the last time point at which each individual would survive under both treatment arms. For individual $1$, although $Y_1^3(1)$ is well-defined because the individual would survive to $12$ months under treatment $z=1$, $\mu(1)$ still uses $Y_1^2(1)$ to ensure comparability with $\mu(0)$. Note that this comparability comes with a caveat. When potential survival times differ between treatment arms, the longer-survival treatment is evaluated before the end of the individual's survival under that treatment. The shorter-survival treatment, by contrast, is evaluated at the final observation before death, when the individual's condition may already have deteriorated.

\end{example}

\begin{example}[Average]\label{exa:averge-w}
 Consider an estimand that averages outcomes across scheduled visits,
$\mu(z) = \frac{1}{2} \Big\{\frac{1}{3} Y_1^0+\frac{1}{3} Y_1^1(z)+\frac{1}{3} Y_1^2(z)\Big\} + \frac{1}{2} \Big\{ \frac{1}{4} Y_2^0+ \frac{1}{4} Y_2^1(z)+ \frac{1}{4} Y_2^2(z) + \frac{1}{4} Y_2^3(z)\Big\},$
where, for each individual, it first averages the potential outcomes over all time points during the guaranteed survival period, and then averages these individual-specific averages over all individuals. This estimand assigns equal weight to each individual, regardless of how long they survive. 

\end{example}

\begin{example}[Cumulative] \label{exa:cumu-mu}
In settings such as health technology assessment and cost-effectiveness analysis, the outcome of interest often accumulates over follow-up. Examples include inpatient bed days, dialysis hours, medical expenditures, symptom burden, toxicity, and QoL burden. Motivated by such settings, consider the cumulative estimand
$\mu(z) = \frac{1}{2} \Big\{Y_1^0 +Y_1^1(z)+Y_1^2(z)\Big\} + \frac{1}{2} \Big\{Y_2^0 +Y_2^1(z)+Y_2^2(z) + Y_2^3(z)\Big\},$
which sums potential outcomes over all time points in the guaranteed survival period. Because each time point enters the sum with equal unit weight, individuals with longer guaranteed-survival periods contribute more to this estimand.

\end{example}

\begin{example}[AUC-based] \label{exa:AUC}
	
For irregular visit schedules, let $\tau_t$ denote the time of visit $t$ since study entry, with $\tau=(\tau_0,\tau_1,\ldots,\tau_{t_{\max}})$ and $\tau_0=0$.	In Figure \ref{fig:toyEx}, $\tau=(0,1/4,1/2,1)$, with time measured in years. A time-weighted cumulative estimand is
$\mu(z) = \frac{1}{2} \Big \{ \frac{1}{4} \frac{Y_1^0 + Y_1^1(z)}{2} + \frac{1}{4} \frac{Y_1^1(z) + Y_1^2(z)}{2} \Big\}
+ \frac{1}{2} \Big \{ \frac{1}{4} \frac{Y_2^0 + Y_2^1(z)}{2} + \frac{1}{4} \frac{Y_2^1(z) + Y_2^2(z)}{2} +  \frac{1}{2} \frac{Y_2^2(z) + Y_2^3(z)}{2}\Big\},$
which incorporates the time intervals to construct an area-under-the-curve (AUC)-based cumulative summary over the guaranteed survival period, as illustrated by the shaded areas in Figure \ref{fig:AUC}.

\begin{figure}[!htbp]
    \centering
    \includegraphics[width=0.7\linewidth]{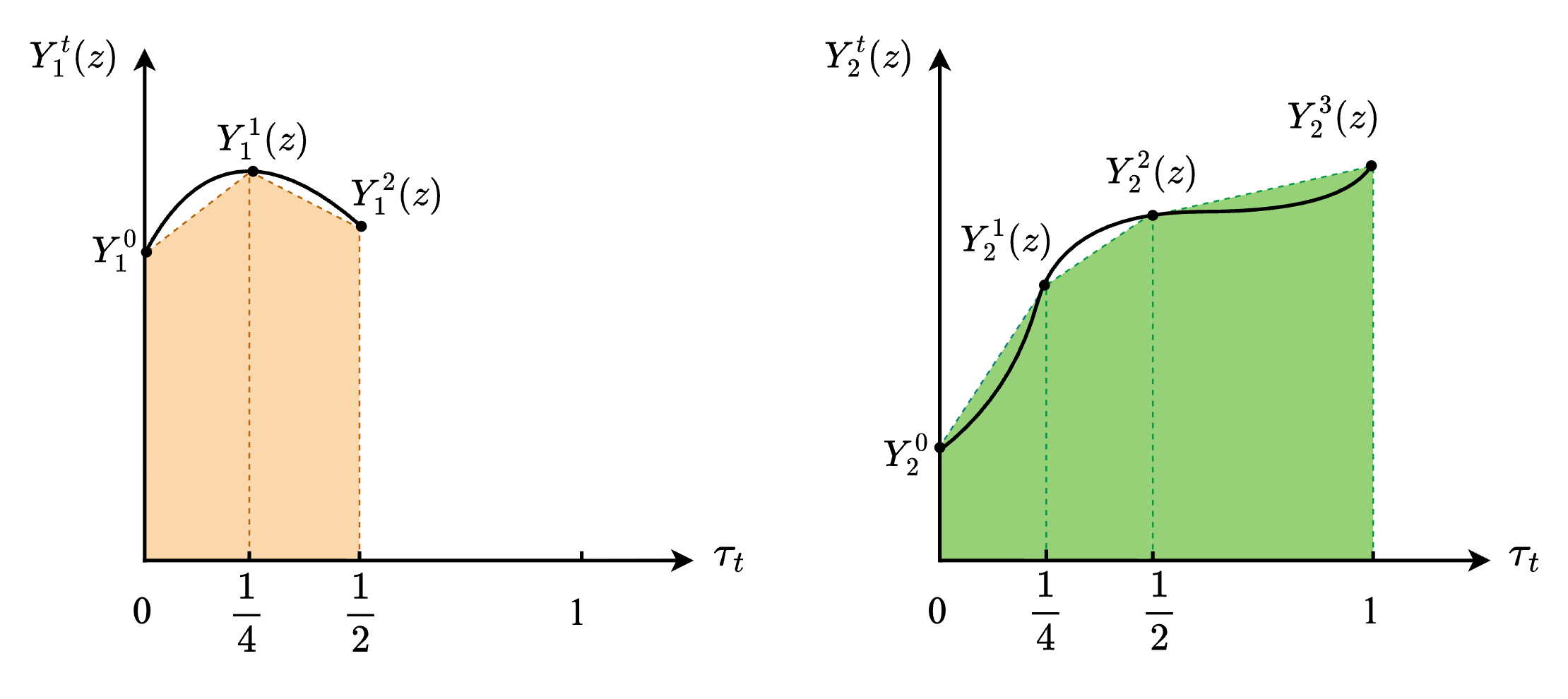}
    \caption{Illustration of the area under the potential outcome curves for individual $1$ (left panel) and individual $2$ (right panel).}
    \label{fig:AUC}
\end{figure}

\end{example}

The preceding examples are special cases of a broader class of while guaranteed-survival estimands. This class provides a unified framework for summarizing well-defined potential outcomes over the guaranteed survival period while allowing the weights to encode different clinical or scientific objectives:
\begin{align}\label{estimand:subA}
\mu(z) :=   \mathbb{E}\Big[  \sum_{r=0}^{T(0) \land T(1)}  w_{T(0) \land T(1)}^r Y^r(z) \Big],
\end{align}
Here \(w_{T(0) \land T(1)}^r\) is the weight assigned to time point \(r=0,1,\ldots,T(0)\land T(1)\). It may depend on \(T(0)\land T(1)\), the last time point at which the individual would be alive under both treatment arms. We define the corresponding treatment contrast as
$
\Delta^{\text{gua}} := \mu(1) - \mu(0),
$
which captures the causal effect of treatment on outcomes over the guaranteed survival period.
In Table~\ref{tab:example-weight}, we present four possible weighting schemes, each corresponding to Examples \ref{exa:exit-w}–\ref{exa:AUC}.

\begin{table}[!htbp]
  \centering
  \caption{Four example types of weighting schemes for the while guaranteed-survival estimand}
  \label{tab:example-weight} 
  \renewcommand{\arraystretch}{1.2}
  \setlength{\tabcolsep}{4pt}
  \small
  \begin{tabular}{ll}
    \toprule
    Type &  Weight \\
    \midrule
    Exit time &  $w_{T(0) \land T(1)}^{\text{exit},r}=
\begin{cases}
1, & r=T(0) \land T(1),\\
0,  & r<T(0) \land T(1)
\end{cases}$ \\
    Average &  $w_{T(0) \land T(1)}^{\text{avg},r}=\frac{1}{\left\{T(0) \land T(1)\right\}+1} $ \\
    Cumulative &  $w_{T(0) \land T(1)}^{\text{cum},r}=1$ \\
    AUC-based &  $\begin{aligned}[t]
w_{T(0) \land T(1)}^{\text{AUC},r}
&=\mathbb{I}(T(0)\land T(1)\geq 1)\\
&\quad \times{}
\begin{cases}
\frac{1}{2}\left[\tau_{T(0) \land T(1)} - \tau_{\{T(0) \land T(1)\}-1}\right], & r=T(0) \land T(1)\\
\frac{1}{2}\left[\tau_{r+1}-\tau_{r-1}\right],  & 0<r< {T(0) \land T(1)}\\
\frac{1}{2}\tau_1, & r=0\\
\end{cases}
\end{aligned}$ \\
    \bottomrule
  \end{tabular}
\renewcommand{\arraystretch}{1}
\end{table}

The following remark relates   our while guaranteed-survival estimands with principal-stratum estimands \citep{robins1986new,rubin2006causal,grossi2023bayesian}.

\begin{remark}[SACE at time $\dot{r}$]\label{remark:SACE}
For a fixed time point $\dot{r}$, when  $w_{T(0)\land T(1)}^r=\frac{\mathbb{I}(T(0)\land T(1)\geq \dot{r})}{\Pr(T(0)\land T(1)\geq \dot{r})} \mathbb{I}(r=\dot{r})$, the contrast of while guaranteed-survival estimands across two treatment arms is reduced to SACE at time $\dot{r}$, i.e., $\Delta^{\text{gua}}=\mathbb{E}[Y^{\dot{r}}(1)-Y^{\dot{r}}(0)|S^{\dot{r}}(1)=S^{\dot{r}}(0)=1]$. From this perspective, the SACE can be viewed as a special case of our while guaranteed-survival estimand with a time-dependent weight that may not be easy to interpret.
\end{remark}

\subsection{A class of {while extended-survival} estimands}\label{sec:Delta-1}

Suppose that an investigator aims to evaluate the potential outcomes during the extended survival period. For an individual with $T(1) > T(0)$, treatment $z=1$ prolongs survival relative to $z=0$. This extended survival period starts after the final time point at which the individual would survive under $z=0$ and ends at the final time point under $z=1$, namely \((T(0), T(1)]\). To assess their potential outcomes under $z=1$ during this period, we summarize them by $\sum_{r=T(0)}^{T(1)}\phi_{T(0),T(1)}^r Y^r(1)$, where the weights $\phi_{T(0),T(1)}^r$ are chosen to summarize this interval and may depend on its start and end points, namely $T(0)$ and $T(1)$. Analogously, for an individual with $T(0)> T(1)$, the potential outcomes under $z=0$ over \((T(1), T(0)]\) can be summarized as $\sum_{r=T(1)}^{T(0)}\phi_{T(1),T(0)}^r Y^r(0)$, where $\phi_{T(1),T(0)}^r$ is defined in the same way as before. Then, to compare the potential outcomes during the extended survival periods under different treatments, we propose the following class of while extended-survival estimands:
\begin{footnotesize}
\begin{align}\label{estimand:delta1}
\mu^{\text{ext}} := \mathbb{E}\left[  \mathbb{I}(T(1) > T(0))\left\{\sum_{r=T(0)}^{T(1)}\phi_{T(0),T(1)}^r Y^r(1)\right\} - \mathbb{I}(T(1)< T(0))\left\{ \sum_{r=T(1)}^{T(0)}\phi_{T(1),T(0)}^r Y^r(0)\right\} \right].
\end{align}
\end{footnotesize}

Define $z_{\min}:=\text{argmin}_z T(z)$ and $z_{\max}:=\text{argmax}_z T(z)$, and we present three examples for the weights $\phi_{T(z_{\min}), T(z_{\max})}^r$, for $r=T(z_{\min}),...,T(z_{\max})$, in Table \ref{tab:example-weight-phi} and illustrate the AUC-based weighting scheme in the following example through the toy example in Figure \ref{fig:toyEx}.

\begin{table}[!htbp]
  \centering
  \caption{Three example types of weighting schemes for the while extended-survival estimand}
  \label{tab:example-weight-phi} 
\renewcommand{\arraystretch}{1.2}
\small
  \begin{tabular}{ll}
    \toprule
    Type & Weight \\
    \midrule
    Average & $\phi_{T(z_{\min}), T(z_{\max})}^{\text{avg},r}= 
\frac{1}{T(z_{\max})-T(z_{\min})} \mathbb{I}(r>T(z_{\min})) $\\
    Cumulative & $\phi_{T(z_{\min}), T(z_{\max})}^{\text{cum},r}=\mathbb{I}(r>T(z_{\min}))$ \\
    AUC-based & $\phi_{T(z_{\min}), T(z_{\max})}^{\text{AUC},r}=
\begin{cases}
\frac{1}{2}\left[\tau_{T(z_{\max})} - \tau_{T(z_{\max})-1}\right], & r=T(z_{\max}),\\
\frac{1}{2}\left[\tau_{r+1}-\tau_{r-1}\right],  & T(z_{\min})<r<T(z_{\max}),\\
\frac{1}{2}\left[\tau_{T(z_{\min})+1} - \tau_{T(z_{\min})}\right],& r=T(z_{\min})\\
\end{cases}$ \\
    \bottomrule
  \end{tabular}
\renewcommand{\arraystretch}{1}
\end{table}

\begin{example}[AUC-based]\label{exa:Delta-AUC}
Figure \ref{fig:AUC-delta}  illustrates the quantities corresponding to the estimands $\Delta^{\text{gua}}$ and $\mu^{\text{ext}}$ using the toy example in Figure \ref{fig:toyEx}. Under the AUC-based weighting scheme $w_{T(0)\land T(1)}^{\text{AUC},r}$, $\Delta^{\text{gua}}=\frac{1}{2}\Delta_1^{\text{gua}}+\frac{1}{2}\Delta_2^{\text{gua}}$, where $\Delta_1^{\text{gua}}=\Big\{\frac{1}{4} \frac{Y_1^0 + Y_1^1(1)}{2} + \frac{1}{4} \frac{Y_1^1(1) + Y_1^2(1)}{2} \Big\} - \Big\{\frac{1}{4} \frac{Y_1^0 + Y_1^1(0)}{2} + \frac{1}{4} \frac{Y_1^1(0) + Y_1^2(0)}{2} \Big\}$ and $\Delta_2^{\text{gua}}=\Big \{ \frac{1}{4} \frac{Y_2^0 + Y_2^1(1)}{2} + \frac{1}{4} \frac{Y_2^1(1) + Y_2^2(1)}{2} +  \frac{1}{2} \frac{Y_2^2(1) + Y_2^3(1)}{2}\Big\}-\Big \{ \frac{1}{4} \frac{Y_2^0 + Y_2^1(0)}{2} + \frac{1}{4} \frac{Y_2^1(0) + Y_2^2(0)}{2} +  \frac{1}{2} \frac{Y_2^2(0) + Y_2^3(0)}{2}\Big\}$ denote, respectively, the areas between the potential outcome curves under treatment and control over the guaranteed survival period for individuals $1$ and $2$. Similarly, under the AUC-based weighting scheme $\phi_{T(z_{\min}),T(z_{\max})}^{\text{AUC},r}$, the while extended-survival estimand is given by $\mu^{\text{ext}}=\frac{1}{2}\mu_1^{\text{ext}}+\frac{1}{2}\mu_2^{\text{ext}}$. As shown in Figure \ref{fig:AUC-delta}, $\mu_1^{\text{ext}}=\frac{1}{2}\frac{Y_1^2(1)+Y_1^3(1)}{2}$ approximates the area under the potential outcome trajectory over the extended survival period from 6 months to 1 year for individual $1$. In contrast, $\mu_2^{\text{ext}}=0$ because, for individual $2$, the extended survival period is empty, as the potential survival time is the same under two treatment arms.

\end{example}

\begin{figure}[!htbp]
    \centering
    \includegraphics[width=0.75\linewidth]{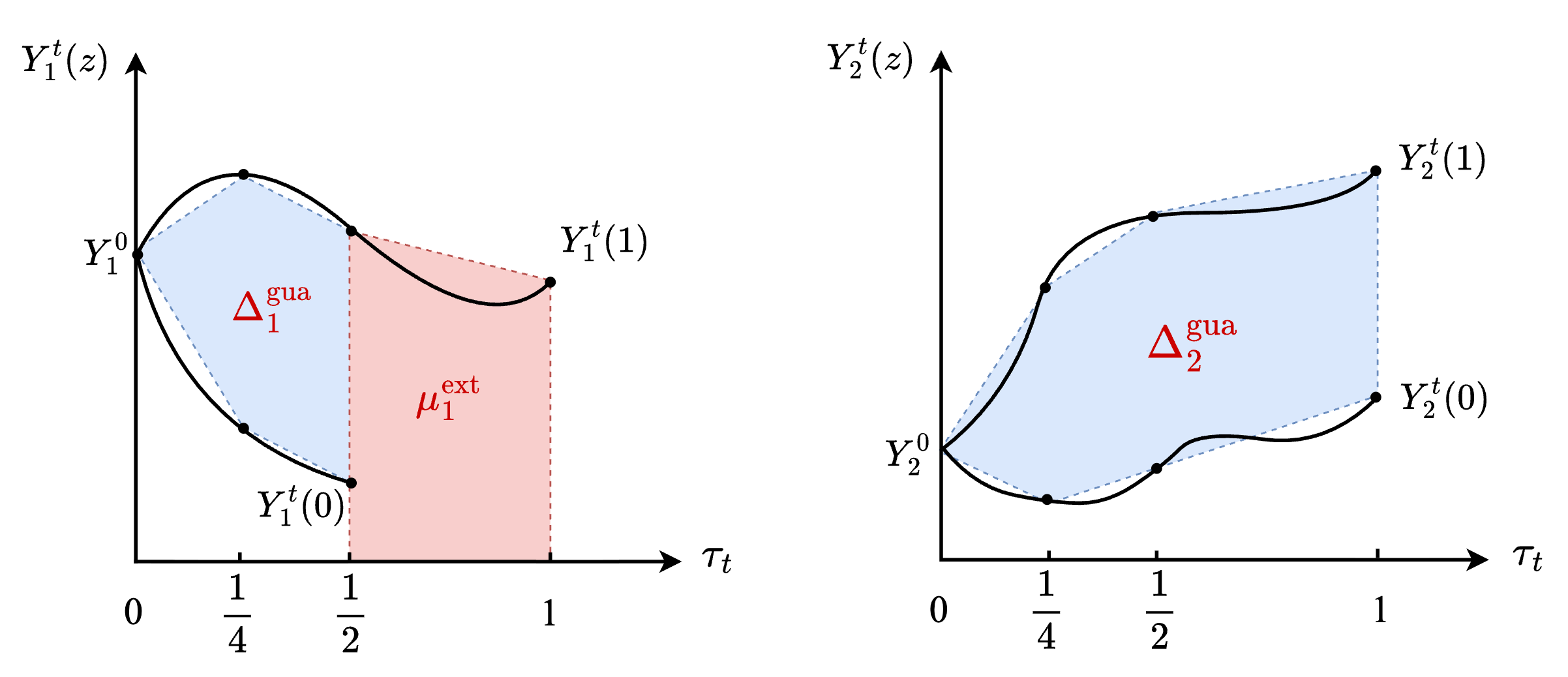}
\caption{Illustration of the areas under the potential outcome curves for two estimands. The blue regions correspond to $\Delta_i^{\mathrm{gua}}$ for $i=1,2$, whereas the red region corresponds to $\mu_1^{\mathrm{ext}}$. For individual 2, the corresponding extended survival region is empty.}
    \label{fig:AUC-delta}
\end{figure}

\section{Identification and estimation of the proposed full-population estimands}\label{sec:subA}

This section shows how the proposed full-population estimands for the two distinct survival periods can be learned from observed longitudinal trial data. We first express the estimands as weighted sums of principal-stratum outcome components, then identify these components using a substitution-variable approach \citep{wang2017identification}, and finally construct plug-in estimators. The assumptions below are used for identification and estimation only and are not part of the estimand definitions. Alternative approaches could be developed under different assumptions, for example by adapting methods based on sensitivity parameters \citep{tong2025semiparametric} or Bayesian principal stratification models \citep{grossi2023bayesian}.

\subsection{Identification of while guaranteed-survival estimands}\label{identify:mu}
To facilitate identification, we rewrite $\mu(z)$ in \eqref{estimand:subA} as
\begin{align}\label{estimand:subA2}
	\mu(z) = \sum_{t=0}^{t_{\max}} \sum_{r=0}^t w_t^r \, \mathbb{E}\!\left[\mathbb{I}\{T(0)\land T(1)=t\} Y^r(z)\right],
\end{align}
where $w_t^r$ denotes the weight $w_{T(0)\land T(1)}^r$ evaluated at $T(0)\land T(1)=t$, and is therefore nonrandom given $t$. Thus, identification of $\mu(z)$ reduces to identifying the component
$\mathbb{E}\!\left[\mathbb{I}\{T(0)\land T(1)=t\}Y^r(z)\right]$ for $0\leq r \leq t \leq t_{\max}$.
By iterated expectation, this component decomposes into a principal-stratum mean potential outcome and the corresponding conditional principal score:
\begin{align*}
&\mathbb{E}\!\left[\mathbb{I}\{T(0)\land T(1)=t\} Y^r(z)\right] \nonumber\\
=& \mathbb{E}\!\left[
\mathbb{E}\!\left\{Y^r(z)\mid T(0)\land T(1)=t,\bar{L}^r(z),L^0\right\}
\Pr\left\{T(0)\land T(1)=t\mid \bar{L}^r(z),L^0\right\}
\right].
\end{align*}

We first introduce the following assumptions commonly used for principal-score identification \citep{wang2017identification, grossi2023bayesian}. 
\begin{assumption}\label{ass:basic-identification}
\begin{enumerate}[label=(\roman*),ref=\theassumption(\roman*),nosep]
\item \textit{Monotonicity.}\label{ass:Monoto-2}
$T(1)\geq T(0)$ almost surely.
\item \textit{(G,L)-ignorability.}\label{ass:S-ig-2}
	 $Z \ind \bar{G}^{t_{\max}}, \bar{L}^{t_{\max}}(z) \mid L^0$ for $z=0,1$.
\item \textit{Positivity.}\label{ass:postivity}
	$0<\Pr(Z=z\mid L^0)<1$ and
	$p_{\bar{L}^r,S^r\mid Z,L^0}(\bar{l}^r,1\mid z,L^0)>0$ almost surely, for any $z=0,1$ and $r=1,\ldots,t_{\max}$.
\item \textit{Y-ignorability.}\label{ass:Y-ig-2}
$Z\ind Y^r(z) \mid  T(z)=t,\bar{L}^r(z), L^0$ and $Z\ind Y^r(1) \mid  T(0)=t,\bar{L}^r(1), L^0$, for $z=0,1$ and $1\leq r\leq t \leq t_{\max}$.
\end{enumerate}
\end{assumption}

For \(r\le t\), we take the conditional principal score \(\Pr\{T(0)\land T(1)=t\mid \bar L^r(z),L^0\}=0\) whenever any component of \(\bar L^r(z)\) is undefined due to death under treatment \(z\). We therefore restrict attention to well-defined histories \(\bar L^r(z)=\bar l^r\), with \(\bar l^r\in\bar{\mathcal L}^r\). Under Assumptions \ref{ass:basic-identification}(i)--(iii), \(
\Pr\{T(0) \land T(1)=t \mid \bar{L}^r(z), L^0\}
\)
is directly identifiable for \(z=0\), because monotonicity reduces \(\{T(0)\land T(1)=t\}\) to \(\{T(0)=t\}\). The case \(z=1\) is more difficult: even under monotonicity, the score \(\Pr\{T(0)\land T(1)=t\mid \bar L^r(1),L^0\}\) conditions on post-treatment covariates under \(z=1\). We therefore impose the following survival principal ignorability assumption.

\begin{assumption}[Survival principal ignorability]\label{ass:surprin-ig}
$S^r(0)
\ind
\bar L^r(1)
\mid
S^r(1)=1,L^0$ for $r=1,\ldots,t_{\max}$.
\end{assumption}

A latent-process rationale for Assumption \ref{ass:surprin-ig} in the cancer-study setting is provided in Remark~\ref{remark:latent-surprin} in Supplementary Material Section~\ref{supple:technical-details}. The identification result for the conditional principal score under \(z=1\) is then given in the following lemma. For notational convenience, we set \(S^{t_{\max}+1}\equiv 0\) throughout.
\begin{lemma}\label{lem::CPS}
\setlength{\abovedisplayskip}{4pt}
\setlength{\belowdisplayskip}{4pt}
\setlength{\abovedisplayshortskip}{4pt}
\setlength{\belowdisplayshortskip}{4pt}
Under Assumptions \ref{ass:basic-identification}(i)--(iii) and \ref{ass:surprin-ig}, for $1\leq r \leq t \leq t_{\max}$, any $\bar{l}^r\in \mathcal{\bar{L}}^r$ and $l^0\in \mathcal{L}^0$, $\Pr\{T(0) \land T(1)=t \mid \bar{L}^r(1)=\bar{l}^r, L^0=l^0\}$ is identified by
\begin{align*}
	\pi_{0/1}^{\,t}(l^0)\,
	p_{S^t\mid\bar L^r,Z,L^0}(1\mid\bar l^r,1,l^0)
	-
	\pi_{0/1}^{\,t+1}(l^0)\,
	p_{S^{t+1}\mid\bar L^r,Z,L^0}(1\mid\bar l^r,1,l^0),
\end{align*}
where \(
{\pi}^r_{0/1}(l^0):={p_{S^r\mid Z,L^0}(1\mid0,l^0)}/{p_{S^r\mid Z,L^0}(1\mid1,l^0)}
\) for $r=1,...,t_{\max}$ and \(\pi_{0/1}^{\,t_{\max}+1}(l^0):=0.\)
\end{lemma}

We now turn to identification of the principal-stratum mean potential outcomes $\mathbb{E}\!\Big\{Y^r(z)\mid T(0)\land T(1)=t, \bar{L}^r(z), L^0\Big\}$,
for $1\leq r \leq t\leq t_{\max}.$ 

\begin{assumption}[G-Markov sufficiency]\label{ass:noYtoS}
$Y^r \ind ~\underline{G}^{r+1} \mid G^r=LL, \bar{L}^r, Z=1, L^0$ for $r=1,\ldots,t_{\max}-1$.
\end{assumption}

Assumptions \ref{ass:Y-ig-2} and \ref{ass:noYtoS} reduce the principal-stratum mean potential outcome to an observed-data outcome regression among individuals in \(G^r=LL\), conditional on \(\bar L^r\), \(Z\), and \(L^0\). Relative to \cite{grossi2023bayesian}, this formulation is less restrictive in that it allows \(Y^r\) and future survival to remain associated through the observed \(\bar L^r\), although the assumptions remain cross-world and untestable in principle \citep{robins2010alternative}. Further discussion is provided in Remark~\ref{remark:gmarkov} in Supplementary Material.

Following \cite{wang2017identification}, let \(A\in L^0\) be a substitution variable and let \(X=L^0\backslash A\). For \((g,z)\in\{(LL,0),(LL,1),(LD,1)\}\), \(\bar l^r\in\bar{\mathcal L}^r\), \(a\in\mathcal A\), and \(x\in\mathcal X\), define \(\mathcal E^r(\bar l^r,z,a,x):=\mathbb E(Y^r\mid S^r=1,\bar L^r=\bar l^r,Z=z,A=a,X=x)\), \(\mathcal M^r(\bar l^r,g,z,a,x):=\mathbb E(Y^r\mid \bar L^r=\bar l^r,G^r=g,Z=z,A=a,X=x)\). With \(l^0=(a,x^\top)^\top\), we write \(\pi_{0/1}^{\,r}(a,x)=\pi_{0/1}^{\,r}(l^0)\). The no-interaction and substitution-relevance assumptions below use \(A\) to recover the outcome regression within \(G^r=LL\), conditional on \(\bar L^r\), \(Z\), and \(L^0\).

\begin{assumption}[Substitution variable conditions]\label{ass:substitution-conditions} For $r=1,...,t_{\max}$ and any $x\in \mathcal{X}$,
\leavevmode 
\begin{enumerate}[label=(\roman*),ref=\theassumption(\roman*),nosep]  
\item \textit{No interaction.}\label{ass:no-inter}
 $\mathcal{M}^r(\bar{l}^r,LD,1,a_1,x) - \mathcal{M}^r(\bar{l}^r,LD,1,a_0,x)
=\mathcal{M}^r(\bar{l}^r,LL,1,a_1,x) -\\ \mathcal{M}^r(\bar{l}^r,LL,1,a_0,x)
= \mathcal{M}^r(\bar{l}^r,LL,0,a_1,x) - \mathcal{M}^r(\bar{l}^r,LL,0,a_0,x)$ for all $a_0,a_1\in \mathcal{A}$ and $\bar{l}^r\in \bar{\mathcal{L}}^r$.
\item \textit{Substitution relevance.}\label{ass:SubRel-2}
 $A \not\perp\!\!\!\perp G^r \mid S^r=1, Z=1, X=x$.
\end{enumerate}
\end{assumption}
We consider Assumption \ref{ass:no-inter} rather than an exclusion restriction–type assumption in \cite{wang2017identification}  to allow for unmeasured confounders between the contemporaneous variables $Y^r$ and $S^r$, such as organ function decline due to treatment toxicity or comorbid conditions in longitudinal clinical trial data.

\begin{theorem}\label{thm:identi:subA} Suppose Assumptions \ref{ass:basic-identification}--\ref{ass:substitution-conditions} hold. Then, for $1\leq r \leq t \leq t_{\max}$ and any $x\in \mathcal{X}$, there exist $a_0, a_1 \in \mathcal{A}$ satisfying \(\pi_{0/1}^{\,r}(a_0,x)\ne\pi_{0/1}^{\,r}(a_1,x)\) such that for every $\bar{l}^r\in \bar{\mathcal{L}}^r$ and $i=0,1$, \(\mathcal{M}^r(\bar{l}^r,LL,1,a_i,x)\) is identified by solving
\begin{align}
&\mathcal{E}^r(\bar{l}^r,1,a_i,x) = \mathcal{M}^r(\bar{l}^r,LL,1,a_i,x) \pi^r_{0/1}(a_i,x) + \mathcal{M}^r(\bar{l}^r, LD,1,a_i,x)(1-\pi^r_{0/1}(a_i,x));\label{eq:identi-subA1}\\
&\mathcal{E}^r(\bar{l}^r,0,a_1,x) - \mathcal{E}^r(\bar{l}^r, 0,a_0,x)
= \mathcal{M}^r(\bar{l}^r,g,1,a_1,x) - \mathcal{M}^r(\bar{l}^r,g,1,a_0,x) \text{, for }g=LL,LD. \label{eq:identi-subA2}
\end{align}
Furthermore,
\begin{align}
&\mathbb{E}\!\left[\mathbb{I}\{T(0)\land T(1)=t\} Y^r(0)\right]
=\mathbb{E}\Big[ \frac{\mathbb{I}(T=t,Z=0)}{\Pr(Z=0|L^0)} Y^r \Big].\label{eq:identi-0}
\\
&\mathbb{E}\!\left[\mathbb{I}\{T(0)\land T(1)=t\} Y^r(1)\right]\nonumber\\
=&\mathbb{E}\Bigg[
\frac{\mathbb{I}(Z=1)}{\Pr(Z=1\mid L^0)}\left\{\mathbb{I}(S^t=1)\pi_{0/1}^t(A,X) - \mathbb{I}(S^{t+1}=1)\pi_{0/1}^{t+1}(A,X)\right\}
\mathcal{M}^r(\bar{L}^r,LL,1,A,X)
\Bigg].\label{eq:identi-1}
\end{align}
\end{theorem}

Note that when \(r=0\), because \(Y^0\in L^0\), under Assumptions \ref{ass:basic-identification}(i)--(iii), 
	\begin{align}\label{eq:identi-baseline}
	\mathbb{E}\!\left[\mathbb{I}\{T(0)\land T(1)=t\}Y^0\right]
	=
	\mathbb{E}\Big[
	\frac{\mathbb{I}(T=t,Z=0)}{\Pr(Z=0\mid L^0)}Y^0
	\Big].
	\end{align}
Together with Theorem \ref{thm:identi:subA}, this identifies all components in \eqref{estimand:subA2}, and hence identifies $\mu(z)$ for $z=0,1$.

\subsection{Identification of while extended-survival estimands}

Under Assumption \ref{ass:Monoto-2}, the while extended-survival estimand in
\eqref{estimand:delta1} reduces to 
\begin{align*}
\mu^{\text{ext}} = \mathbb{E}\left[ \mathbb{I}(T(1) > T(0)) \sum_{r=T(0)}^{T(1)} \phi_{T(0),T(1)}^r Y^r(1)\right].
\end{align*}
A complication is that the start and end points of the extended survival period, \(T(0)\) and \(T(1)\), are random, so the weight \(\phi_{T(0),T(1)}^r\) is random as well. For the cumulative and AUC-based weighting schemes, this dependence can be handled by rewriting \(\mu^{\text{ext}}\) as the difference between two weighted summaries under treatment:
\begin{align}\label{estimand:delta1-2}
\mu^{\text{ext}} = \mathbb{E}\left[\left\{\sum_{r=0}^{T(1)} \phi_{T(1)}^r Y^r(1)\right\} - \left\{\sum_{r=0}^{T(0)} \phi_{T(0)}^r Y^r(1)\right\}\right],
\end{align}
where $\phi_{T(z)}^r$ denotes a weight that may depend on $T(z)$. The choices of \(\phi_{T(z)}^r\) corresponding to the cumulative and AUC-based weighting schemes are given in Supplementary Table~\ref{tab:example-weight-phi-re}. Under the AUC-based scheme, this representation equals the difference between the areas under the trajectory over \([0,T(1)]\) and \([0,T(0)]\), as illustrated by the red shaded region in Figure \ref{fig:AUC-delta}.

Similar to Section \ref{identify:mu}, we rewrite $\mu^{\text{ext}}$ in \eqref{estimand:delta1} as:
\begin{align}\label{estimand:delta1-3}
\mu^{\text{ext}} 
= \sum_{t=0}^{t_{\max}}\sum_{r=0}^t \phi_t^r \, \mathbb{E}\!\left[\mathbb{I}\{T(1)=t\}Y^r(1)\right]
- \sum_{t=0}^{t_{\max}}\sum_{r=0}^t \phi_t^r \, \mathbb{E}\!\left[\mathbb{I}\{T(0)=t\}Y^r(1)\right],
\end{align}
where $\phi_t^r$ denotes the weight $\phi_{T(z)}^r$ evaluated at $T(z)=t$, and is therefore nonrandom conditional on $T(z)=t$.
Under Assumptions \ref{ass:basic-identification}(ii)--(iv), the first term in \eqref{estimand:delta1-3} is identified as
$
\mathbb{E}\!\left[\mathbb{I}\{T(1)=t\}Y^r(1)\right]
= \mathbb{E}\!\left[ \frac{\mathbb{I}(T=t, Z=1)}{\Pr(Z=1\mid L^0)} Y^r \right],
$
while identification of the second term in \eqref{estimand:delta1-3} follows from results in Section \ref{identify:mu}.

So far, we have shown that under the cumulative and AUC-based weighting schemes 
$\phi_{T(0),T(1)}^r$, the estimand $\mu^{\text{ext}}$ admits the decomposition in 
\eqref{estimand:delta1-2} and is therefore identifiable. More generally, we say a weighting scheme 
$\{\phi_{T(0),T(1)}^r: r=T(0),\ldots,T(1)\}$ is \emph{decomposable} if there exist weights 
$\{\phi_{T(0)}^r: r=0,\ldots,T(0)\}$ and 
$\{\phi_{T(1)}^r: r=0,\ldots,T(1)\}$ such that 
$\phi_{T(1)}^r=\phi_{T(0)}^r$ for $r=0,\ldots,T(0)-1$, and
$\phi_{T(0),T(1)}^r=\phi_{T(1)}^r-\phi_{T(0)}^{r}\mathbb{I}\{r=T(0)\}$
for $r=T(0),\ldots,T(1)$. Under any decomposable weighting scheme, $\mu^{\text{ext}}$ admits the representation in \eqref{estimand:delta1-2}.

\begin{remark}
Not all weighting schemes $\phi_{T(0),T(1)}^r$ are decomposable, and decomposability is required for the identification result above. For example, the average weighting scheme $\phi_{T(z_{\min}), T(z_{\max})}^{\mathrm{avg},r}$ is not decomposable because, for $T(0)<r\leq T(1)$, the weight depends on both $T(0)$ and $T(1)$. Therefore, the identification results above do not directly apply to $\mu^{\text{ext}}$ under the average weighting scheme.
\end{remark}

\subsection{Estimation}\label{sec:subA-est}

The previous sections established that $\mu(z)$ and $\mu^{\text{ext}}$ are identifiable under the stated identification assumptions. 
When $L^0$ is high-dimensional and continuous, however, additional model parameterization is required. This step is not merely computational. It translates the identification formulas into estimators while preserving the monotonicity and survival-ordering restrictions implied by the assumptions in Section \ref{identify:mu}. We impose the following constraints:
\begin{align}
&\Pr(S^r(0)=1 \mid X,A) \leq \Pr(S^r(1)=1 \mid X,A), \quad r=1,\ldots,t_{\max}; \label{cons:mon}\\
&\Pr(S^{r+1}(z)=1 \mid X,A) \leq \Pr(S^r(z)=1 \mid X,A), \quad r=1,\ldots,t_{\max}-1,\; z=0,1; \label{cons:surviv}\\
&\mathcal{E}^r(\bar{l}^r,1,a,x) - \mathcal{E}^r(\bar{l}^r,0,a,x)
\text{ is bounded as a function of } a \in \mathcal{A}, \text{ for } x \in \mathcal{X}, \bar{l}^r\in \bar{\mathcal{L}}^r, \; r=1,\ldots,t_{\max}. \label{cons:3}
\end{align}
Constraint \eqref{cons:mon} follows from Assumption \ref{ass:Monoto-2}, and constraint \eqref{cons:3} is implied by \eqref{eq:identi-subA1}--\eqref{eq:identi-subA2}. In addition, because $S^r$ is non-increasing in $r$, constraint \eqref{cons:surviv} is required in addition to the estimation procedure of \citet{wang2017identification}.

We next propose parameterizations that satisfy constraints \eqref{cons:mon}--\eqref{cons:3} and are compatible with the identifiability assumptions in Section \ref{identify:mu}. The parameterizations are chosen to encode the restrictions needed for the observed-data likelihood to align with the principal-stratum identification argument. For $r=1,\ldots,t_{\max}$,

\begin{itemize}
\item[$\mathcal{P}_1:$] $\mathcal{M}^r(\bar{L}^r,G^r,Z,A,X)=m(\bar{L}^r,G^r,Z,A,X;\alpha^r)$, where $\alpha^r$ is an unknown finite-dimensional parameter and $G^r$ takes values in $\{LL,LD\}$;

\item[$\mathcal{P}_2:$] $\Pr(S^r(1)=1 \mid S^{r-1}(1)=1, L^0)=l_1(L^0;\beta^r)$, where $\beta^r$ is an unknown finite-dimensional parameter;
\item[$\mathcal{P}_3:$] $\frac{\Pr(S^r(0)=1 \mid S^{r-1}(0)=1, L^0)}{\Pr(S^r(1)=1 \mid S^{r-1}(1)=1,L^0)}=l_{0/1}(L^0;\gamma^r)$, where $\gamma^r$ is an unknown finite-dimensional parameter;
\item[$\mathcal{P}_4:$] $ \Pr(Z=1\mid L^0) = l_z(L^0;\theta)$, where $\theta$ is an unknown finite-dimensional parameter.
\end{itemize}
For simulation and real data applications, we use linear models for $m(\cdot)$ and logistic regression models for $l_1(\cdot)$, $l_{0/1}(\cdot)$, and $l_z(\cdot)$, all without interaction terms. 
The observed-data modeling constraints implied by these parameterizations are provided in Remark~\ref{remark:obs-model-constraints} in Supplementary Material Section~\ref{supple:technical-details}. We note that this specification of $\mathcal{P}_3$ implies that the relative survival advantage of active treatment over control becomes more pronounced over time, consistent with the trend in our real data in Section~\ref{sec:real-data}. In Remark~\ref{remark:sur-model-restriction} of Supplementary Material
Section~\ref{supple:technical-details}, we also provide an alternative
parameterization that does not impose this additional restriction.

Let $\widehat{\alpha}^r$, $\widehat{\beta}^r$,
$\widehat{\gamma}^r$, and $\widehat{\theta}$ denote the
corresponding estimators, obtained by least squares for $\alpha^r$ and maximum likelihood for $\beta^r$, $\gamma^r$, and $\theta$. Using \eqref{eq:identi-0}, \eqref{eq:identi-1}, and \eqref{eq:identi-baseline}, we define the following plug-in quantities:
\begin{align*}
\widehat{Q}_{0,t}^r
&= \mathbb{P}_n \left\{\frac{\mathbb{I}(T=t, Z=0)}{1-l_z(L^0;\widehat{\theta})}Y^r \right\},
\widehat{Q}_{1,t}^r =
\begin{cases}
\widehat{Q}_{0,t}^0, & r=0,\\
\mathbb{P}_n \left\{\frac{\mathbb{I}( Z=1)}{l_z(L^0;\widehat{\theta})} \widehat{D}^t m(\bar{L}^r,LL,1,A,X;\widehat{\alpha}^r) \right\}, & r\ge 1,
\end{cases}
\end{align*}
where $\mathbb{P}_n$ is the empirical average, $\widehat{D}^t=\mathbb{I}(S^t=1)\widehat{\pi}^t_{0/1}(A,X)-\mathbb{I}(S^{t+1}=1)\widehat{\pi}^{t+1}_{0/1}(A,X)$, and $\widehat{\pi}^t_{0/1}(A,X)$ is defined as $\prod_{k=1}^t l_{0/1}(L^0;\widehat{\gamma}^k) $ if $t\leq t_{\max}$ and $0$ if $t=t_{\max}+1$.
Finally, the estimators for $\mu(z)$ and $\mu^{\text{ext}}$, for $z=0,1$, are given by
$\widehat{\mu}(0)= \sum_{t=0}^{t_{\max}} \sum_{r=0}^t w_t^r \widehat{Q}_{0,t}^r, 
\widehat{\mu}(1)= \sum_{t=0}^{t_{\max}} \sum_{r=0}^t w_t^r \widehat{Q}_{1,t}^r,$ and 
$$\widehat{\mu}^{\text{ext}} 
= \sum_{t=0}^{t_{\max}}\sum_{r=0}^t \phi_t^r \mathbb{P}_n \left\{\frac{\mathbb{I}(T=t, Z=1)}{l_z(L^0;\widehat{\theta})} Y^r \right\}
- \sum_{t=0}^{t_{\max}}\sum_{r=0}^t \phi_t^r \widehat{Q}_{1,t}^r.
$$

\section{Population-level summaries of separable effects}\label{sec:se}

The while guaranteed-survival and while extended-survival estimands target the full population, but they remain cross-world because they depend on the joint distribution of $T(1)$ and $T(0)$. In this section, we revisit the separable effects introduced in Section \ref{sec:separable}, which conceptually decompose treatment into two components. We develop a new class of marginal separable effect estimands that extend the CSE to the full population.

\subsection{A class of marginal separable estimands}\label{sec:alt-estimand}
Inspired by the CSE and applying an analogous construction to Section \ref{sec:MargEstimand}, we define a class of marginal separable estimands for the full population
\begin{align}\label{estimand:CSE}
\Gamma(z_Y,z_S):=\mathbb{E}\left[\sum_{r=0}^{T(z_S)} \psi_{T(z_S)}^r Y^r(z_Y,z_S)\right],
\end{align}
where $\psi_{T(z_S)}^r$ is a weight, for $r=0,1,...,T(z_S)$, that may depend on $T(z_S)$, the last time point at which the individual would be alive under $Z_S=z_S$. In general, $\Gamma(z_Y,z_S)$ collapses the longitudinal potential outcomes under treatment components $(Z_Y, Z_S)=(z_Y, z_S)$ during the time period $[0,T(z_S)]$, and the contrast $\Delta^{\text{sep}}(z_S):=\Gamma(1,z_S)-\Gamma(0,z_S)$ is the corresponding average causal effect of treatment $Z_Y$ on outcomes had all individuals been assigned $Z_S=z_S$. Unlike the cross-world survival time \(T(0)\land T(1)\), we are now considering a single-world quantity \(T(z_S)\). The event \(\{T(z_S)=t\}\) could therefore, in principle, be learned from a future four-arm trial of the treatment components described in Section \ref{sec:separable}.

\begin{table}[!htbp]
  \centering
  \caption{Four example types of weighting schemes for the marginal separable effect estimands}
  \label{tab:example-weight-se} 
  \renewcommand{\arraystretch}{1.2}
  \small
  \begin{tabular}{lll}
    \toprule
    Type &  Weight \\
    \midrule
    Exit time &  $\psi_{T(z_S)}^{\text{exit},r}=
\begin{cases}
1, & r=T(z_S),\\
0,  & r<T(z_S)
\end{cases}$ \\
    Average &  $\psi_{T(z_S)}^{\text{avg},r}=\frac{1}{T(z_S)+1} $ \\
    Cumulative &  $\psi_{T(z_S)}^{\text{cum},r}=1$ \\
    AUC-based &  $\psi_{T(z_S)}^{\text{AUC},r}=\mathbb{I}(T(z_S)\geq 1) \times
\begin{cases}
\frac{1}{2}\left[\tau_{T(z_S)} - \tau_{T(z_S)-1}\right], & r=T(z_S),\\
\frac{1}{2}\left[\tau_{r+1}-\tau_{r-1}\right],  & 0<r< {T(z_S)},\\
\frac{1}{2}\tau_1, & r=0
\end{cases}$ \\
    \bottomrule
  \end{tabular}
\renewcommand{\arraystretch}{1}
\end{table}

The weighting scheme in \eqref{estimand:CSE} may be chosen to reflect different clinical goals, and Table~\ref{tab:example-weight-se} presents four illustrative examples.   In addition, similar to Remark \ref{remark:SACE}, we provide the following remark to illustrate the relationship between our marginal separable estimands and the CSE. 

\begin{remark}[CSE at time $\dot{r}$]\label{remark:CSE} For a fixed time point $\dot{r}$, when $\psi_{T(z_S)}^r=\frac{\mathbb{I}(T(z_S)\geq \dot{r})}{\Pr(T(z_S)\geq \dot{r})}\mathbb{I}(r=\dot{r})$, we have $\Gamma(z_Y,z_S)=\mathbb{E}[Y^{\dot{r}}(z_Y,z_S)\mid S^{\dot{r}}(z_S)=1]$ and $\Delta^{\text{sep}}(z_S)$ coincides with the CSE at time $\dot{r}$. Thus, the CSE can be interpreted as a special case of our marginal separable effect estimand.
\end{remark}

\subsection{Identification}
To identify $\Gamma(z_Y,z_S)$ for $z_Y,z_S\in \{0,1\}$, we first rewrite \eqref{estimand:CSE} as
\begin{align}\label{estimand:CSE-2}
\Gamma(z_Y, z_S) =  \sum_{t=0}^{t_{\max}} \sum_{r=0}^t \psi_t^r 
\mathbb{E}\left[ \mathbb{I}\{T(z_S)=t\}Y^r(z_Y,z_S)\right],
\end{align}
where $\psi_t^r$ denotes the weight $\psi_{T(z_S)}^r$ evaluated at $T(z_S)=t$. 
Thus, identification of $\Gamma(z_Y,z_S)$ reduces to identifying 
$\mathbb{E}[\mathbb{I}\{T(z_S)=t\}Y^r(z_Y,z_S)]$ for $0 \le r \le t \le t_{\max}$.

Motivated by existing work on conditional separable effects, and in addition to Assumption \ref{ass:Z_S-DCC} on the treatment component $Z_S(\mathcal{C})$, we impose the following dismissible component conditions for $Z_Y$ and an S-Markov sufficiency assumption under the four-arm trial $\mathcal{C}$.

\begin{assumption}
[$Z_Y$ dismissible component condition] \label{ass:DCC_LongSE-2} For $r=1,\ldots,t_{\max}$,
\begin{itemize}
    \item[(1)] $Z_Y(\mathcal{C}) \ind S^r(\mathcal{C}) \mid Z_S(\mathcal{C}),\, S^{r-1}(\mathcal{C})=1,\, \bar{L}^r(\mathcal{C}),\, L^0$;
    \item[(2)] $Z_Y(\mathcal{C}) \ind L^r(\mathcal{C}) \mid Z_S(\mathcal{C}),\, S^{r-1}(\mathcal{C})=1,\, \bar{L}^{r-1}(\mathcal{C}),\, L^0$.
\end{itemize}
\end{assumption}

\begin{assumption}[S-Markov sufficiency]\label{ass:S-markov}

\(
\underline{S}^{r+1}(\mathcal{C}) \ind Y^{r}(\mathcal{C})
\mid Z_Y(\mathcal{C}),\, Z_S(\mathcal{C}),\, S^{r}(\mathcal{C})=1,\, \bar{L}^{r}(\mathcal{C}),\, L^0,
\) for $r=1,\ldots,t_{\max}-1$.
\end{assumption}

Assumptions \ref{ass:DCC_LongSE-2} and \ref{ass:S-markov} can be assessed from the causal graph under trial $\mathcal{C}$ in Figure \ref{fig:Longi-SE-2}(b) using d-separation. Assumption \ref{ass:DCC_LongSE-2} reflects the conceptual decomposition of treatment, under which the outcome-related component $Z_Y$ does not directly affect survival or survival-related time-varying covariates $\bar{L}^{t_{\max}}$. For example, in a cancer clinical trial, these time-varying covariates often characterize tumor progression or disease status, which $Z_Y$ is not expected to directly affect. Similarly to Assumption \ref{ass:noYtoS} in Section \ref{identify:mu}, Assumption \ref{ass:S-markov} may be plausible in the
absence of directed paths from $Y^r(\mathcal{C})$ to
$(L^t(\mathcal{C}),S^t(\mathcal{C}))$ for any $r<t$. The validity of this assumption depends on the context, in particular the length of the follow-up window. For example, quality of life is unlikely to have direct effects on death over short time periods, say months or a few years, although these variables likely have strong common causes. Over longer time horizons, however, such direct effects become more plausible. We further impose the following standard weak ignorability assumption.

\begin{assumption}[Ignorability]\label{ass:ranIgn}
$ Z \ind \bar{Y}^{t_{\max}}(z),\bar{S}^{t_{\max}}(z), \bar{L}^{t_{\max}}(z) \mid L^0$.
\end{assumption}

\begin{theorem}\label{thm:identi-SE}
Suppose Assumptions \ref{ass:Z_yPI}, \ref{ass:Z_S-DCC}, \ref{ass:MTA}, \ref{ass:postivity}, and \ref{ass:DCC_LongSE-2} - \ref{ass:ranIgn} hold. For \(r=0\) and \(t=0,\ldots,t_{\max}\), the components $\mathbb{E}[\mathbb{I}\{T(z_S)=t\}Y^r(z_Y, z_S)]$ in \eqref{estimand:CSE-2} are identified as 
\begin{align*}
\mathbb{E}\!\left[\mathbb{I}\{T(z_S)=t\}Y^0\right]
= \mathbb{E}\left[\frac{\mathbb{I}(T=t, Z=z_S)}{\Pr(Z=z_S\mid L^0)}Y^0\right].
\end{align*}
For $1\leq r \leq t\leq t_{\max}$,
\begin{align}\label{eq:iden-SE}
\mathbb{E}\left[\mathbb{I}\{T(z_S)=t\}Y^r(z_Y, z_S)\right]
= \mathbb{E}\left[ \mathbb{E}\left\{ \mathbb{I}(T=t)\, \mathbb{E}(Y^r \mid S^r=1, \bar{L}^r, Z=z_Y, L^0)\Big | Z=z_S, L^0\right\} \right].
\end{align}
\end{theorem}

Theorem \ref{thm:identi-SE} identifies all components in \eqref{estimand:CSE-2}, and hence identifies $\Gamma(z_Y, z_S)$.

\subsection{Connection to while guaranteed-survival and while extended-survival estimands}\label{sec:connect}

In general, the while guaranteed-survival and while extended-survival estimands differ from the marginal separable effect estimands. Conceptually, the former characterizes the effects by partitioning the survival time into two distinct periods. In contrast, the marginal separable effects assess the effect of $Z_Y$ on outcomes while holding the survival-related component fixed at $Z_S=z_S$, over $[0, T(z_S)]$. Nevertheless, these estimands can be connected under additional assumptions.

In particular, under Assumptions \ref{ass:Z_yPI}--\ref{ass:MTA} and \ref{ass:Monoto-2}, both $\mu(z)$ and $\Gamma(z_Y,0)$ summarize the potential outcomes during the time period $[0,T(0)]$. However, 
$\Delta^{\text{sep}}(0)$ isolates the effect of $Z_Y$ on $\{Y^r: r=1,...,T(0)\}$ over $[0,T(0)]$ by fixing $z_S=0$. In contrast, $\Delta^{\text{gua}}$ captures not only the effect of $Z_Y$ on $\{Y^r: r=1,...,T(0)\}$, but also the effect of $Z_S$ on $\{Y^r: r=1,...,T(0)\}$ mediated through $\{L^r: r=1,...,T(0)\}$. This distinction is similar in spirit to our discussion of the SACE and CSE in Section \ref{sec:separable}. Therefore, under the following assumption that rules out any additional effect of $Z_S$ on $Y^t$ that does not pass through $\bar{S}^t$, we have $\Delta^{\text{gua}}=\Delta^{\text{sep}}(0)$ under appropriate choices of the weighting schemes.

\begin{assumption}[$Z_S$ partial isolation] \label{ass:Z_sPI}
For $z_Y\in \{0,1\}$, $Y^t(z_Y,1)=Y^t(z_Y,0)$ if ${S}^t(1)={S}^t(0)=1$.
\end{assumption}

On the other hand, under Assumptions \ref{ass:Monoto-2} and \ref{ass:Z_sPI}, $\Gamma(1,0)=\mathbb{E}\left[\sum_{r=0}^{T(0)} \psi_{T(0)}^r Y^r(1,1)\right]$. Hence, with appropriate choices of the weighting schemes, $\Gamma(1,1)-\Gamma(1,0)$ could summarize the potential outcomes under $(z_Y,z_S)=(1,1)$ over \((T(0),T(1)]\), aligning with the quantity targeted by $\mu^{\text{ext}}$. Then the following proposition specifies the conditions under which $\Delta^{\text{gua}}=\Delta^{\text{sep}}(0)$ and $\mu^{\text{ext}}=\Gamma(1,1)-\Gamma(1,0)$. When choosing appropriate weighting schemes to establish the equivalence, we focus on fixed-index weights $w_t^r$ in \eqref{estimand:subA2}, $\phi_t^r$ in \eqref{estimand:delta1-3}, and $\psi_t^r$ in \eqref{estimand:CSE-2}.

\begin{prop}\label{prop:connection}
Under Assumptions \ref{ass:Z_yPI}, \ref{ass:MTA}, \ref{ass:Monoto-2}, and \ref{ass:Z_sPI}, 
if $w_t^r=\psi_t^r$, then $\Delta^{\text{gua}}=\Delta^{\text{sep}}(0)$; 
if $\phi_t^r=\psi_t^r$, then $\mu^{\text{ext}} = \Gamma(1,1) - \Gamma(1,0)$.
\end{prop}

\subsection{Estimation}\label{sec:estimator-SE}

Let \(\Lambda_t^r(z_Y,z_S)\) denote the corresponding component identified in Theorem \ref{thm:identi-SE} for \(0\le r \le t\). The influence function in Theorem \ref{thm:triplyRobust-SE} of Supplementary Material Section~\ref{supple:technical-details} motivates estimation under \(\mathcal P_4\), defined in Section \ref{sec:subA-est}, and two additional model classes. For \(r=1,\ldots,t_{\max}\),
\begin{itemize}
\item[$\mathcal{P}_5$:] $\mathbb{E}[Y^r \mid S^r=1,\bar{L}^r,Z, L^0]=\kappa(\bar{L}^r,Z,L^0;\eta^r)$, where $\eta^r$ is an unknown finite-dimensional parameter;

\item[$\mathcal{P}_6$:] 
$p_{L^r\mid S^{r-1},\bar L^{r-1},Z,L^0}(l^r \mid 1,\bar L^{r-1},Z,L^0) = l_L(l^r,\bar L^{r-1},Z,L^0;\rho_L^r)$, 
and $p_{S^r\mid S^{r-1},\bar L^r,Z,L^0}(1\mid1,\bar L^r,Z,\\L^0) = l_S(\bar L^r,Z,L^0;\rho_S^r)$, 
where $\rho^r=((\rho_L^r)^\top,(\rho_S^r)^\top)^\top$ is an unknown finite-dimensional parameter.
\end{itemize}

Solving the corresponding estimating equation under \(\mathcal P_4\cap\mathcal P_5\cap\mathcal P_6\) yields a triply robust estimator of \(\Lambda_t^r(z_Y,z_S)\), consistent if any two model classes are correctly specified. Because \(\mathcal P_6\) requires modeling $p_{L^r\mid S^{r-1},\bar L^{r-1},Z,L^0}$, we use the simpler estimator obtained by omitting the \(\mathcal P_6\)-dependent terms. For $r=1,\ldots,t$,
\begin{align*}
\widehat{\Lambda}_t^r(z_Y,z_S) 
= \mathbb{P}_n \left\{ 
\frac{\mathbb{I}(T=t, Z=z_S)}{z_S l_z(L^0;\widehat{\theta})+(1-z_S)(1-l_z(L^0;\widehat{\theta}))} 
\kappa(\bar{L}^r,z_Y,L^0;\widehat{\eta}^r)
\right\}.
\end{align*}
This simplified estimator is consistent when both \(\mathcal P_4\) and \(\mathcal P_5\) are correctly specified.
For \(r=0\), \(\widehat{\Lambda}_t^0(z_Y,z_S)\) uses the same inverse-probability weighted form, with the outcome-regression term replaced by \(Y^0\).

The resulting estimator is
$
\widehat{\Gamma}(z_Y,z_S) 
= \sum_{t=0}^{t_{\max}} \sum_{r=0}^t \psi_t^r \widehat{\Lambda}_t^r(z_Y,z_S).
$
In simulations and the application, we use the same model for $l_z(\cdot)$ as in Section \ref{sec:subA-est} and specify \(\kappa(\cdot)\) using linear models without interaction terms.

Simulation studies assessing the finite-sample performance of the proposed estimators in Sections~\ref{sec:subA-est} and~\ref{sec:estimator-SE} are reported in Supplementary Material Section~\ref{supple:simulation}.

\section{The Southwest Oncology Group Trial}\label{sec:real-data}

We illustrate our approaches by reanalyzing a randomized phase III prostate cancer trial conducted by the Southwest Oncology Group (SWOG) between October 1999 and January 2003, comparing Docetaxel plus Estramustine (DE) with Mitoxantrone plus Prednisone (MP) in men with metastatic androgen-independent prostate cancer \citep{petrylak2004docetaxel}. 
The dataset created by \citet{ding2011identifiability} contains 487 men aged $47-88$ years, with $258$ randomly assigned to  DE ($Z=1$) and $229$ to MP ($Z=0$). Baseline covariates ($L^0$) include age, race, type of prognosis, bone pain, performance status, and baseline QoL. Health-related QoL, scored from 0 to 100, was assessed at 3 months, 6 months, and 1 year among survivors. We set $\tau=(0,1/4,1/2,1)$, corresponding to the elapsed times in years from baseline to the scheduled assessment times, and derive a time-varying disease progression indicator $L^t$ $(t=1,2,3)$ from the recorded progression date.

\begin{figure}[!htbp]
    \centering
    \includegraphics[width=0.55\linewidth]{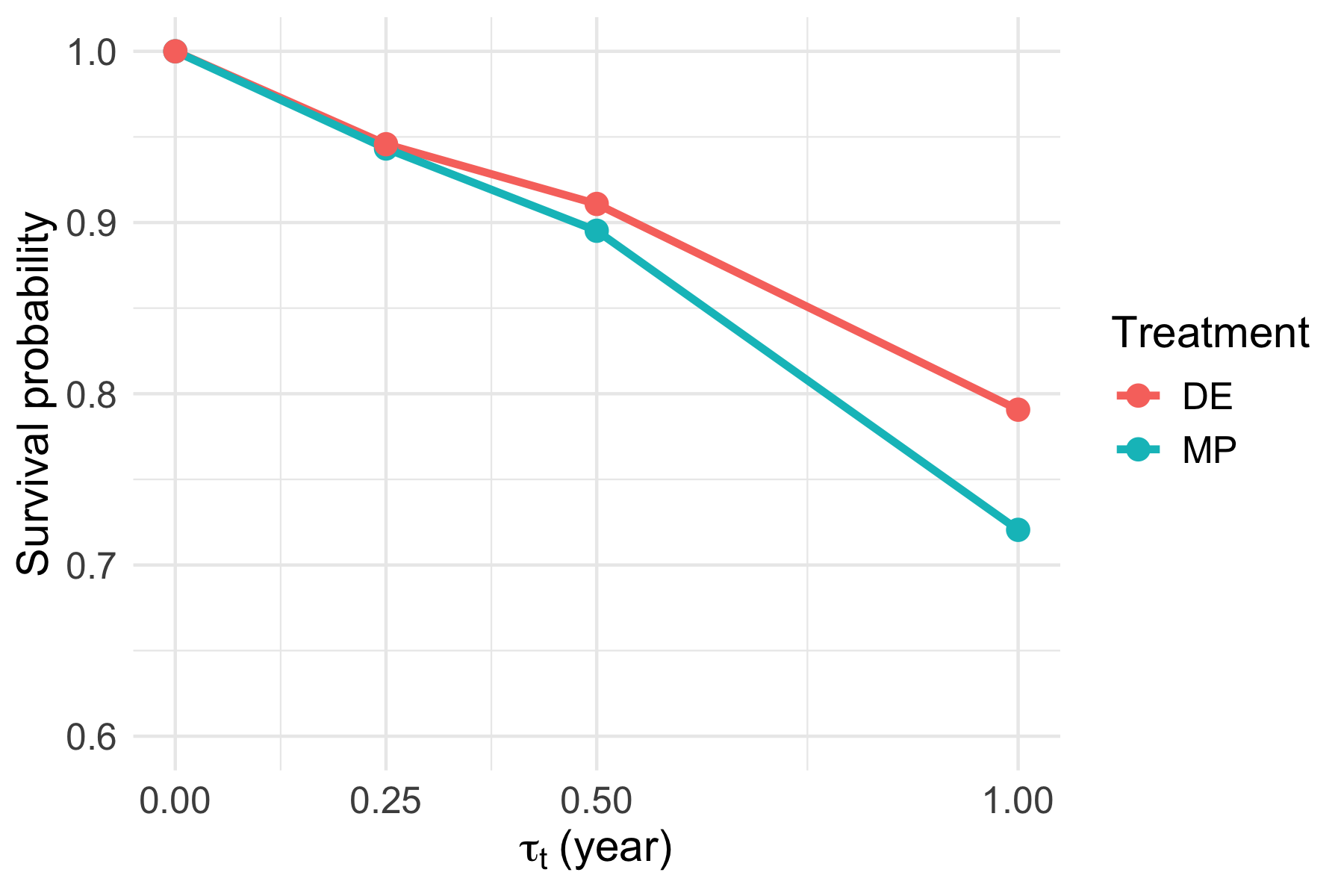}
    \caption{Estimated survival curve under two chemotherapies.}
    \label{fig:data-survival}
\end{figure}

Following \cite{park2024proximal}, we impute missing QoL values among survivors using k-nearest neighbors (KNN) imputation \citep{troyanskaya2001missing}. Death is treated as a truncating event rather than as a missing QoL value. Figure \ref{fig:data-survival} shows that patients receiving DE have higher survival probabilities than those receiving MP, consistent with prior findings reported in \citet{petrylak2004docetaxel}. However, \citet{petrylak2004docetaxel} reported higher rates of grade 3 or 4 neutropenic fevers, nausea and vomiting, and cardiovascular events among patients receiving DE than among those receiving MP. The survival benefit and toxicity profile together motivate an analysis of treatment effects on changes in QoL from baseline.

Several analyses have compared changes in QoL between treatment groups among the always-survivor group at the 6-month or 1-year follow-up \citep{ding2011identifiability,wang2017identification,yang2018using}. Following \cite{ding2011identifiability} and \cite{wang2017identification}, we use baseline QoL as a substitution variable to estimate the SACE at each follow-up time, denoted by SACE(\(t\)) for \(t=1,2,3\). For uncertainty quantification, we draw 500 bootstrap samples from the unimputed dataset and repeat KNN imputation and estimation in each sample.  As shown in the left panel of Figure \ref{fig:data-sace}, the point estimates differ in direction across endpoints: the 6-month and 1-year endpoints suggest a QoL benefit of DE relative to MP among the corresponding always-survivor group, whereas the 3-month endpoint suggests the opposite, although all corresponding 95\% confidence intervals cover zero. However, SACE(\(t\)) should not be compared directly across \(t\), because as shown in the right panel of Figure \ref{fig:data-sace}, the always-survivor group changes and shrinks over time. Focusing on these typically healthier subgroups also limits the interpretability and practical relevance of the findings for the full trial population.

\begin{figure}[!htbp]
    \centering
    \includegraphics[width=0.95\linewidth]{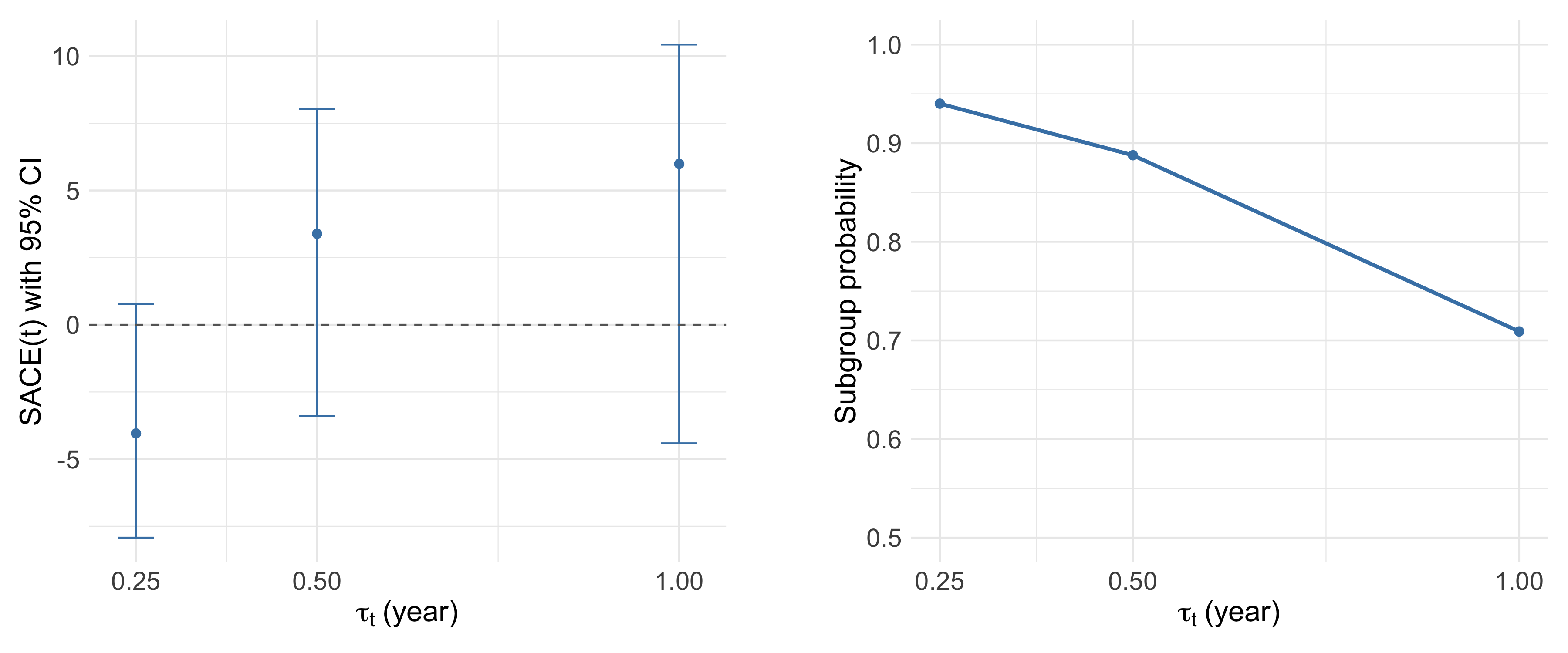}
    \caption{The left panel shows the estimated SACE(t) at three time points together with the corresponding 95\% confidence intervals, and the right panel shows the estimated subgroup probability for the always-survivor group at each time point.}
    \label{fig:data-sace}
\end{figure}

We then apply the while guaranteed-survival estimands to evaluate the effect of chemotherapy on changes in QoL during the guaranteed survival period, again using baseline QoL as the substitution variable. For comparison, we report the while-alive contrast \(\lambda(1)-\lambda(0)\), which is commonly used in longitudinal analyses but does not admit a causal interpretation in this setting. Details of the estimation procedure for \(\lambda(z), z=0,1\) are provided in Supplementary Material Section~\ref{supple:while-alive}.

Table \ref{tab:comparison-guaranteedperiod} reports estimates of \(\lambda(1)-\lambda(0)\) and \(\Delta^{\text{gua}}\) under several weighting schemes. The while-alive point estimates are small and positive, with 95\% confidence intervals covering zero. The while guaranteed-survival point estimates are also positive, with 95\% confidence intervals covering zero, but the exit-time estimate is close to the commonly used \([-5,5]\) benchmark for a clinically meaningful difference in QoL \citep{cocks2012evidence}. This variation is informative because the weighting schemes target different clinical questions.

Under the exit-time weighting scheme, the while guaranteed-survival estimate is much larger than the corresponding while-alive estimate and is smaller than the SACE at 1 year. This pattern is consistent with the concern that while-alive summaries compare outcomes at different exit times under DE and MP, which may attenuate the estimated effect when DE prolongs survival; the SACE at 1 year, by contrast, targets a healthier always-survivor subgroup. Comparing \(\Delta^{\text{gua}}\) across weighting schemes suggests that the estimated benefit of DE relative to MP depends on the clinical goal of interest: DE appears more favorable when emphasis is placed on exit-time or end-of-life outcomes, whereas the difference is small when the focus is average QoL over time or the longitudinal trajectory.

\begin{table}[!htbp]
    \centering
    \caption{Estimates for the contrasts $\lambda(1)-\lambda(0)$, $\Delta^{\text{gua}}$ and $\Delta^{\text{sep}}(0)$ under different weighting schemes. The bootstrap standard errors and quantile-based confidence intervals are derived from $500$ bootstrap samples}
    \small
    \begin{tabular}{clccc}
        \toprule
        Approach & Weight & Point estimate & Bootstrapped SE &  95\% CI \\
        \midrule
        \multirow{3}{*}{\makecell{While-alive\\
        $\lambda(1)-\lambda(0)$}} 
        & Exit time & 0.91 & 2.70 & ($-$4.22,  6.33)\\
        & Average & 0.40 & 1.30 & ($-$2.41, 2.85)\\
        & AUC-based & 1.05 & 1.41 & ($-$1.92,  3.72)\\
        \midrule
        \multirow{4}{*}{\makecell{While guaranteed-survival\\
        $\Delta^{\text{gua}}$}} 
        & SACE at 1 year &  5.99 & 3.84 & ($-$4.41, 10.43) \\
        & Exit time & 4.28 & 3.10 & ($-$3.17, 9.13) \\
        & Average & 1.47 & 1.45  & ($-$2.41,  3.62)  \\
        & AUC-based & 1.67 & 1.53  &  ($-$2.33, 3.74)\\
        \midrule
        \multirow{4}{*}{\makecell{Marginal separable effect\\
        $\Delta^{\text{sep}}(0)$}} 
        & CSE at 1 year & 1.19 & 3.30 & ($-$4.36, 8.22) \\
        & Exit time & 0.68 & 2.51& ($-$3.87, 6.02)\\
        & Average & $-$0.66 & 1.24 & ($-$3.21, 1.44)\\
        & AUC-based & $-$0.54 & 1.32 & ($-$3.29, 1.71)\\
    \bottomrule
    \end{tabular}
    \label{tab:comparison-guaranteedperiod}
\end{table}

We next consider marginal separable effects by conceptualizing the chemotherapy regimen as having a survival-related component \(Z_S\), which affects survival through antitumor activity, and a QoL-related component \(Z_Y\), which affects QoL more directly through symptom relief or treatment toxicity \citep{stensrud2023conditional}. This decomposition is conceptual rather than experimentally implemented in the SWOG trial, so the estimates should be interpreted through the assumptions of the separable effects framework. We estimate \(\Gamma(z_Y,z_S)\) by adjusting for baseline covariates and the time-varying indicator of cancer progression, and Table \ref{tab:comparison-guaranteedperiod} reports \(\Delta^{\text{sep}}(0)\).

The estimates of \(\Delta^{\text{sep}}(0)\) are close to zero, with confidence intervals covering zero, and are smaller than the corresponding \(\Delta^{\text{gua}}\) estimates. This pattern is consistent with the distinction in Section \ref{sec:connect}: \(\Delta^{\text{gua}}\) may include pathways through which the survival-related component affects QoL via cancer progression, whereas \(\Delta^{\text{sep}}(0)\) isolates the QoL-related component while fixing \(Z_S=0\). The larger \(\Delta^{\text{gua}}\) estimates are therefore compatible with an indirect QoL benefit through the survival-related component of DE, consistent with the stronger antitumor role of DE and the more palliative role historically associated with MP \citep{petrylak2004docetaxel,tannock1996chemotherapy}.

Finally, we evaluate QoL during the extended survival period under DE relative to MP using the AUC-based weighting scheme, which accounts for the unequal spacing of the follow-up assessments. Table \ref{tab:ext_qol_survival} reports the restricted survival-time contrast, the while extended-survival estimate, and the corresponding cumulative QoL values during the extended survival period. Results suggest that DE increased 1-year restricted survival time by \(0.05\) years, or about \(0.6\) months, relative to MP. Over this additional survival time, \(\mu^{\mathrm{ext}}=-0.62\) indicates a cumulative QoL loss of \(0.62\) QoL score-years relative to baseline. After adding back baseline QoL, the corresponding AUC-based cumulative QoL was \(1.60\) QoL score-years. Because the estimated survival gain is small, normalizing this cumulative QoL change by the additional survival time would be unstable and could imply a large average QoL decrement over the added time. We therefore interpret these results on the cumulative scale: the point estimate suggests some QoL decline during the modest additional survival period under DE, but the confidence intervals are wide and include values compatible with little change or improvement.

According to Table \ref{tab:ext_qol_survival}, the marginal separable contrast, \(\Gamma(1,1)-\Gamma(1,0)\), is larger than the corresponding estimate of \(\mu^{\mathrm{ext}}\). This agrees with the interpretations of the two estimands: \(\mu^{\mathrm{ext}}\) summarizes QoL during the extended survival period under the full treatment regimens, whereas \(\Gamma(1,1)-\Gamma(1,0)\) isolates the contribution of the survival-related component while holding the QoL-related component fixed. Thus, the larger marginal separable contrast may reflect that, by improving tumor control, the survival-related component of DE may reduce cancer progression and thereby indirectly improve QoL.

\begin{table}[htbp]
		\centering
		\caption{Summary measures for evaluating QoL during the extended survival time under DE, using AUC-based weighting. Entries are reported as estimates \((95\%\ \text{bootstrap CI})\)}
		\small
		\setlength{\tabcolsep}{4pt}
		\begin{tabular}{@{}M{0.40\textwidth} C{0.23\textwidth} C{0.29\textwidth}@{}}
			\toprule
			Measure & Primary estimate & \begin{tabular}[c]{@{}c@{}}Corresponding AUC-based\\ cumulative QoL\end{tabular} \\
			\midrule
			Increase in restricted 1-year survival time under DE relative to MP (years)
			&  \makecell{\(0.05\)\;\((0,\;0.10)\)} & \textemdash \\
			\addlinespace[2pt]
			While extended-survival \(\mu^{\mathrm{ext}}\)
			& $-$0.62 ($-$1.39, 2.13) & \(1.60\; (0,\;5.74)\) \\
			\addlinespace[2pt]
			Marginal separable effect \(\Gamma(1,1)-\Gamma(1,0)\)
			& \(1.60\; (0.69,\;2.70)\) & \(3.82\; (0.69,\; 6.82)\) \\
		\bottomrule
	\end{tabular}
	\label{tab:ext_qol_survival}
\end{table}

\section{Discussion}

When death truncates a longitudinal outcome, a causal estimand must specify not only the target population and treatment contrast, but also which part of each potential outcome trajectory is being summarized. This paper gives a unified theory and methodology for such estimands. The results can be used to formulate full-population summaries and to extend two widely used classes of causal estimands: principal stratification, represented by the SACE, and separable effects. The prostate cancer application demonstrates how these choices lead to estimands that answer distinct causal questions.

Several limitations remain. The while guaranteed-survival and while extended-survival estimands are full-population estimands but remain cross-world, and their identification relies on assumptions such as monotonicity, principal ignorability, and substitution-variable conditions. In addition, our identification strategy assumes that the outcome process does not directly affect subsequent survival after conditioning on the measured history. Some of these restrictions may be relaxable with additional auxiliary information. Proxy-variable approaches \citep{park2024proximal} may help extend the separable-effect formulation to settings with unmeasured survival-outcome confounding. Separately, alternative substitution-variable strategies could exploit baseline or early post-baseline variables that help distinguish latent survival histories, thereby targeting components such as \(\mathbb{E}\{Y^r(z)\mid T(0)\land T(1)=t\}\) under different identifying restrictions. Future work may also develop sensitivity analyses for these assumptions and extend the framework to continuous-time measurements, competing intercurrent events, and dynamic treatment regimes.

\section*{Acknowledgements}

We thank Stijn Vansteelandt for helpful discussions that motivated this work, and Zhen Luo for insightful comments that led to the development of the while extended-survival estimands. We are also grateful to Gregory Chen, Vanessa Didelez, Guowen Huang, Fabrizia Mealli, Thomas Richardson, James Robins, Jessica Young, Xin Zhang, and Xiao-Hua Zhou for helpful comments. Zhao was partially supported by a CANSSI Ontario Postdoctoral Fellowship in Statistical Sciences. Stensrud and Wang would like to thank the Isaac Newton Institute for Mathematical Sciences, Cambridge, for support and hospitality during the programme Foundations of causal inference, where work on this paper was completed. 
{We used OpenAI ChatGPT to assist with grammar checks and code organization. All substantive content, analyses, code, and conclusions were developed, reviewed, and verified by the authors.}

\section*{Data availability}
Code for the simulation studies and real-data analysis is available at \url{https://github.com/Ruixuan-Zhao/marginal-estimands-truncated-by-death}. The real data are not publicly available, so a simulated dataset is provided for running the real-data analysis code.

\begingroup
\setlength{\bibsep}{0pt}
\putbib
\endgroup
\end{bibunit}

\clearpage
\begin{bibunit}

\begin{center}
{\Large\bfseries Supplementary Material for\\[0.35em]
``Causal Inference for All: Marginal Estimands for Outcomes Truncated by Death''\par}
\vspace{1em}
Ruixuan Zhao\textsuperscript{1}, Mats Stensrud\textsuperscript{2}, and Linbo Wang\textsuperscript{1,3}\par
\vspace{0.6em}

\textsuperscript{1}Department of Computer and Mathematical Sciences, University of Toronto Scarborough\\
\textsuperscript{2}Institute of Mathematics, \'Ecole Polytechnique F\'ed\'erale de Lausanne\\
\textsuperscript{3}Department of Statistical Sciences, University of Toronto
\end{center}

\vspace{1em}
\noindent\textbf{Abstract.}
This supplement provides additional material for the manuscript
``Causal Inference for All: Marginal Estimands for Outcomes Truncated by Death.''
Section~\ref{supple:existing-estimands} reviews existing estimand strategies for longitudinal outcomes truncated by death.
Section~\ref{supple:technical-details} gives additional technical details for the identification and estimation results in the main text.
Section~\ref{supple:while-alive} describes identification and estimation of while-alive estimands, and Section~\ref{supple:simulation} reports additional simulation details.

\bigskip

\suppsection{Existing estimands for longitudinal outcomes truncated by death}{supple:existing-estimands}

This section provides additional discussion on existing estimands for longitudinal outcomes truncated by death. Section \ref{sec:related} focuses on the while-alive strategy because it is most closely connected to the proposed estimands. Here, we briefly review other commonly used strategies.

Early approaches focused on naive comparisons at a given time point, such as $\mathbb{E}[Y^t(z)\mid S^t(z)=1]$, or equivalently $\mathbb{E}[Y^t(z)\mid T(z)\ge t]$ \citep{kurland2009longitudinal}. These estimands compare outcomes among individuals who would be alive under treatment $z$ at time $t$, but contrasts across treatment arms generally involve different survivor subpopulations and are therefore subject to selection bias.

A related line of work defines estimands under hypothetical interventions that eliminate the truncating event, sometimes referred to as hypothetical estimands \citep{ich2019addendum}. The hypothetical estimand can be written as $\mathbb{E}(Y^t(z,s^t=1))$, which quantifies the effect of treatment on the outcome at time $t$ if the individual had survived. However, this intervention is usually infeasible and therefore has limited practical relevance \citep{young2020causal}.

Composite estimands have also been used to summarize joint survival and outcome behavior at the population level. For example, \citet{diehr1995including} considered estimands, $\mathbb{E}[S^t(z)\mathbb{I}(Y^t(z)\geq y)]$, defined as the probability of being alive at a given time with QoL exceeding a prespecified threshold, directly addressing questions such as the chance of being alive at three years with acceptable QoL. Related formulations assign a prespecified value to the outcome following truncation \citep{lu2025estimating} or rank death relative to observed outcomes \citep{lok2010long, xiang2023survival}. However, the selection of a prespecified value is not always straightforward. For example, under a composite estimand such as $\mathbb{E}[S^t(z)Y^t(z)]$ \citep{lu2025estimating}, QoL after death is often set to zero to represent the worst state, yet it is not always clear whether death itself or survival with severely compromised QoL should be regarded as worse \citep{rosenbaum2006comment}. Moreover, the composite outcome arguably does not shed light on the mechanistic effect of treatment on the outcome of interest that is not mediated by the intercurrent event.

More recently, to test the null hypothesis of no treatment effect on longitudinal outcomes truncated by death, \citet{baklicharov2025weakening} proposed a time-dependent Pairwise Last Observation
Time (PLOT) estimand, which contrasts outcomes between pairs of individuals receiving different
treatments at a common time point where both individuals would survive.
In particular, they considered two independent random individuals, where one is assigned to treatment $z=1$ with observed $(\bar{Y}^{t_{\max}}(1),T(1))$ and the other (indexed with $*$) is assigned to $z=0$ with observed $(\bar{Y}^{t_{\max}}_*(0), T_*(0))$. Then, for a fixed time point $t$, the PLOT estimand is defined as $\mathbb{E}[Y^{T(1)\land T_*(0)\land t}(1) - Y_*^{T(1)\land T_*(0)\land t}(0)]$, where the expectation is taken over all independently drawn pairs of individuals assigned to different treatments. As discussed further in Remark~\ref{remark:PLOT}, this estimand may still be subject to residual selection bias. \citet{baklicharov2025weakening} further showed that this bias can be eliminated under additional conditions at the cost of stronger assumptions.

\begin{sidewaystable}[!htbp]
	\centering
	\caption{Summary of estimands for longitudinal data with outcomes truncated by death}
	\label{tab:sum-est}
	\small
	\setlength{\tabcolsep}{4pt}
	\begingroup
	\renewcommand{\arraystretch}{1.2}
	\begin{tabular}{M{0.10\textheight} M{0.25\textheight} M{0.35\textheight} M{0.23\textheight}}
		\toprule
		\textbf{Strategy} & \textbf{Description} & \textbf{Estimand} & \textbf{Caveats}  \\
		\midrule
		Naive & (Dynamic) stratification by survival status or follow-up time \citep{kurland2009longitudinal} & \estcell{\(\mathbb{E}[Y^t(z)\mid S^t(z)=1]\ \text{or}\ \mathbb{E}[Y^t(z)\mid T(z)\geq t]\)}   & Subject to selection bias \\
		\midrule
		While-alive
		& While-alive estimand \citep{wei2023properties}  & \estcell{\(\lambda(z)\) in \eqref{esimand-while-alive}} & Subject to selection bias \\
		\midrule
		Hypothetical & Hypothetical estimand under the elimination of death \citep{young2020causal} & \estcell{\(\mathbb{E}[Y^t(z,{s}^t=1)]\)}   & Assume values for undefined outcomes \\
		\midrule
		\multirow[c]{4}{=}[0.6\baselineskip]{Composite} & Joint estimand of survival and outcome \citep{diehr1995including}  &  \estcell{\(\mathbb{E}[S^t(z)\mathbb{I}(Y^t(z)\geq y)]\)} & Threshold-based dichotomization loses information  \\
		& Composite estimand \citep{lu2025estimating} & \estcell{\(\mathbb{E}[S^t(z)Y^t(z)]\)} & Assigning a prespecified value can be controversial  \\
		\midrule
		\multirow[c]{4}{=}[0.6\baselineskip]{Principal stratum} & Principal stratum \citep{robins1986new,rubin2006causal, frangakis2002principal} & \estcell{\(\mathbb{E}[Y^t(z)\mid G^t=LL]\ \text{or}\ \mathbb{E}[Y^t(z)\mid T(1)\geq t,\ T(0)\geq t]\)}  &  Defined on an unknown subpopulation and inapplicable to individuals outside \\[10pt]
		& Principal stratum in longitudinal setting \citep{grossi2023bayesian} & \estcell{\(\mathbb{E}[Y^r(z)\mid T(1)\geq t,\ T(0)\geq t]\ \text{for } r\leq t\)}  & Subgroup compositions change over time \\
		\midrule
		\multirow[c]{4}{=}[0.6\baselineskip]{Others} & PLOT estimand \citep{baklicharov2025weakening}  & \estcell{\(\mathbb{E}\big[Y^{T(1)\land {T}_*(0)\land t}(1)\big]\ \text{vs.}\ \mathbb{E}\big[{Y}_*^{T(1)\land {T}_*(0)\land t}(0)\big]\)} &  Subject to residual selection bias unless further particular assumptions hold \\
		& Conditional separable effect \citep{stensrud2023conditional}  & \estcell{\(\mathbb{E}[Y^t(z_Y,z_S)\mid S^t(z_S)=1]\)}  & Defined on a subpopulation  \\
		\bottomrule
	\end{tabular}
	\endgroup
\end{sidewaystable}

For ease of reference, Table \ref{tab:sum-est} summarizes the strategies reviewed above and in Sections 2.2--2.4. For each estimand, the table provides the mathematical formulation, a brief description, and a caveat highlighting key limitations, with representative references. 

\begin{remark}[Comparison to the PLOT Estimand] \label{remark:PLOT}
	Our estimand in Example \ref{exa:exit-w} shares a feature with the PLOT estimand in \cite{baklicharov2025weakening}: both ensure that comparisons across treatment arms are made at a common time point. However, the PLOT estimand $$\mathbb{E}\left [Y^{T(1)\land T_*(0)\land t}(1)-Y_*^{T(1)\land T_*(0)\land t}(0)\right ]$$ may exhibit residual selection bias relative to the treatment effect at the same time point (i.e.,\\ \(\mathbb{E}\Big [Y^{T(1)\land{T}(0)\land t}(1) -\allowbreak {Y}^{T(1)\land {T}(0)\land t}(0)\Big ]\)). The effects we consider, like $\mathbb{E}\left [Y^{T(1)\land{T}(0)}(1) - {Y}^{T(1)\land {T}(0)}(0)\right ]$, avoid such bias.
	
	We use the toy examples in Figure \ref{fig:toyEx} to illustrate that the PLOT estimand does not always preserve the null of no treatment effect. Suppose that the potential outcomes take the values shown in the left panel of Figure \ref{fig:PLOT_est}, then the null that $Y_i^t(1)=Y_i^t(0)$ for any $t\leq T_i(1)\land T_i(0)$ and $i=1,2$ holds. The contrast of our estimand in Example \ref{exa:exit-w} is equal to zero, aligning with the null. However, the PLOT estimand with $t=3$ is equal to $\frac{1}{2}\left \{Y_2^2(1)-Y_1^2(0)\right \}+\frac{1}{2}\left\{Y_1^3(1)-Y_2^3(0)\right\}=-20$, which implies that this test is not valid in this scenario. Conversely, a zero PLOT estimand does not necessarily imply the absence of a treatment effect. For example, as shown in the right panel of Figure \ref{fig:PLOT_est}, the null fails to hold, but the PLOT estimand with $t=3$ is equal to zero. 
\end{remark}

\begin{figure}[!htbp]
	\centering
	\includegraphics[width=0.65\linewidth]{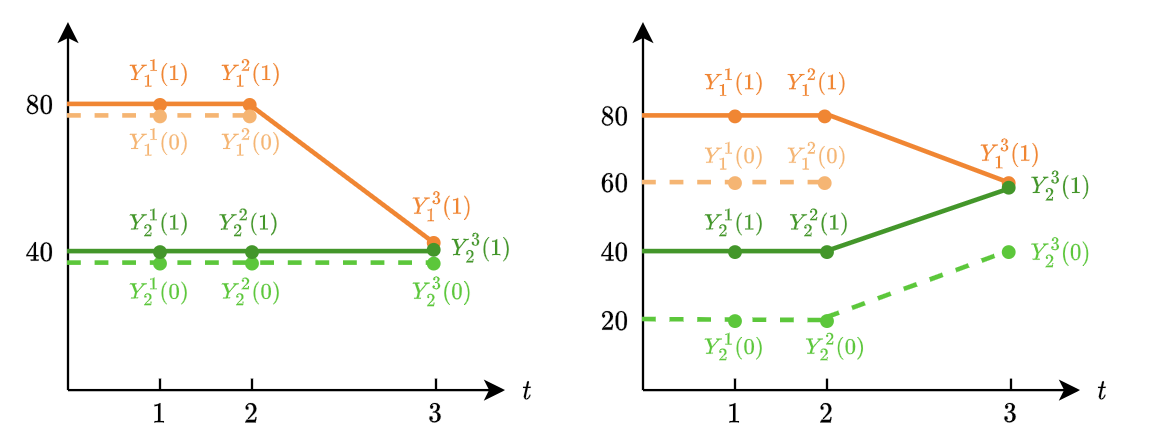}
	\caption{Two examples with the values of the potential outcomes to illustrate the potential pitfalls of the PLOT estimand. For individual $1$, the orange solid line denotes the longitudinal potential outcomes under treatment $z=1$, and the orange dashed line denotes the potential outcomes under $z=0$. For individual $2$, the green solid and dashed lines denote the potential outcomes under $z=1$ and $z=0$, respectively.}
	\label{fig:PLOT_est}
\end{figure}

\suppsection{Additional Technical Details for Identification and Estimation in Sections \ref{sec:subA} and \ref{sec:se}}{supple:technical-details}

\begin{table}[!htbp]
  \centering
  \caption{Choices of $\phi_{T(z)}^r$ in \eqref{estimand:delta1-2} corresponding to the cumulative and AUC-based weighting schemes listed in Table \ref{tab:example-weight-phi}}
  \label{tab:example-weight-phi-re} 
\renewcommand{\arraystretch}{1.1}
\small
  \begin{tabular}{ll}
    \toprule
    Type & Weight \\
    \midrule
    Cumulative & $\phi_{T(z)}^{\text{cum},r}=1$ \\
    AUC-based & $\phi_{T(z)}^{\text{AUC},r}=\mathbb{I}(T(z)\geq 1)\times
\begin{cases}
\frac{1}{2}(\tau_{T(z)} - \tau_{{T(z)}-1}), & r=T(z),\\
\frac{1}{2}(\tau_{r+1}-\tau_{r-1}),  & 0<r< T(z),\\
\frac{1}{2}\tau_1, & r=0
\end{cases}$ \\
    \bottomrule
  \end{tabular}
\renewcommand{\arraystretch}{1}
\end{table}

\begin{remark}[A latent-process rationale for survival principal ignorability]\label{remark:latent-surprin}
	In the cancer study, $L^r(z)$ can be an indicator of potential disease progression by time $r$ under treatment $z$. Suppose that disease progression and death may arise through three underlying cancer-related processes. Let $T_{\mathrm{H}}$ denote the failure time associated with a treatment-sensitive harmful process, which may lead to disease progression or death under control but can be eliminated by active treatment. Let $T_{\mathrm{P}}$ denote the failure time associated with a treatment-resistant disease progression process, which may lead to disease progression under either control or active treatment. Let $T_{\mathrm{D}}$ denote the failure time associated with a treatment-resistant death process, which may lead to death under either control or active treatment. Then, we have the following latent process formulation, for $r=1,\ldots,t_{\max}$,
	\begin{align*}
		L^r(0) 
		&= \mathbb{I}\left\{\min(T_{\mathrm{H}},T_{\mathrm{P}},T_{\mathrm{D}})\le r\right\},
		&
		L^r(1) 
		&= \mathbb{I}\left\{\min(T_{\mathrm{P}},T_{\mathrm{D}})\le r\right\},\\
		S^r(0) 
		&= \mathbb{I}\left\{\min(T_{\mathrm{H}},T_{\mathrm{D}})> r\right\},
		&
		S^r(1) 
		&= \mathbb{I}\left\{T_{\mathrm{D}}> r\right\}.
	\end{align*}
	For \(z=0,1\), the expressions for \(L^r(z)\) apply only when \(S^{r-1}(z)=1\); otherwise, \(L^r(z)\) is undefined due to death. If $T_{\mathrm H}$ is independent of $(T_{\mathrm P},T_{\mathrm D})$ conditional on baseline covariates $L^0$, then the above latent process formulation implies Assumption~\ref{ass:surprin-ig}. 
\end{remark}

\begin{remark}[Discussion on G-Markov sufficiency]\label{remark:gmarkov}
	As shown in Figure \ref{fig:Longi-SE-2}(a), Assumption \ref{ass:noYtoS} may be plausible when there are no directed paths from $Y^r$ to $S^t$, and all back-door paths between $Y^r$ and $S^t$ are blocked by $L^0$, $\bar{L}^r$ and $\bar{S}^r$ for any $r<t$. It may hold for QoL outcomes if changes in QoL primarily reflect current disease progression, and subsequent survival is driven by disease progression rather than QoL itself. The assumption may also be plausible for medical cost outcomes, which are broadly better interpreted as a marker of disease severity and healthcare utilization than as a direct cause of subsequent survival.
\end{remark}

\begin{proof}[Proof of Lemma \ref{lem::CPS}]
	For continuous or mixed covariates, the same argument applies with sums and probability masses replaced by the corresponding integrals and densities. For $z=1$ and $t<t_{\max}$, under Assumptions \ref{ass:Monoto-2},
	\begin{align*}
	&\Pr\left(T(0)\land T(1)=t \mid \bar{L}^r(1)=\bar{l}^r, L^0=l^0\right) = \Pr\left(T(0)=t \mid \bar{L}^r(1)=\bar{l}^r, L^0=l^0\right)\\
	=& \Pr\left(S^t(0)=1 \mid \bar{L}^r(1)=\bar{l}^r, L^0=l^0\right)-\Pr\left(S^{t+1}(0)=1 \mid \bar{L}^r(1)=\bar{l}^r, L^0=l^0\right)
	\end{align*}
Then, we have
	\begin{align*}
		&\Pr\left(S^t(0)=1 ,\bar{L}^r(1)=\bar{l}^r \mid L^0=l^0\right)
		=\sum_{\bar{l}^t \backslash \bar{l}^r} \Pr \left(S^t(0)=1 , \bar{L}^t(1)=\bar{l}^t \mid L^0=l^0\right)\\
	=&\sum_{\bar{l}^t \backslash \bar{l}^r} \Pr \left(S^t(0)=1 , S^t(1)=1, \bar{L}^t(1)=\bar{l}^t \mid L^0=l^0\right)\\
	=& \sum_{\bar{l}^t \backslash \bar{l}^r} \Pr \left(S^t(0)=1 \mid S^t(1)=1, \bar{L}^t(1)=\bar{l}^t, L^0=l^0\right) \Pr \left( S^t(1)=1, \bar{L}^t(1)=\bar{l}^t \mid L^0=l^0\right)\\
    =& \sum_{\bar{l}^t \backslash \bar{l}^r} \Pr \left(S^t(0)=1 \mid S^t(1)=1, L^0=l^0\right) \Pr \left( S^t(1)=1, \bar{L}^t(1)=\bar{l}^t \mid L^0=l^0\right)\\
    =& \Pr \left(S^t(0)=1 \mid S^t(1)=1, L^0=l^0\right) \Pr \left( S^t(1)=1, \bar{L}^r(1)=\bar{l}^r \mid L^0=l^0\right)\\
    =& \frac{\Pr \left(S^t(0)=1, S^t(1)=1 \mid L^0=l^0\right)}{\Pr \left(S^t(1)=1 \mid L^0=l^0\right)} \Pr \left( S^t(1)=1, \bar{L}^r(1)=\bar{l}^r \mid L^0=l^0\right)\\
    =&\frac{\Pr \left(S^t=1\mid Z=0, L^0=l^0\right)}{\Pr \left(S^t=1 \mid Z=1, L^0=l^0\right)} \Pr \left( S^t=1, \bar{L}^r=\bar{l}^r \mid Z=1, L^0=l^0\right),
	\end{align*}
where the second equation follows from Assumption \ref{ass:Monoto-2}, the fourth equation follows from Assumption \ref{ass:surprin-ig}, and the last equation follows from Assumptions \ref{ass:S-ig-2} and \ref{ass:postivity}. It follows from Assumption \ref{ass:S-ig-2} that
\begin{align*}
	&\Pr\left(S^t(0)=1 \mid \bar{L}^r(1)=\bar{l}^r, L^0=l^0\right)\\
	=& \frac{\Pr\left(S^t(0)=1 ,\bar{L}^r(1)=\bar{l}^r \mid L^0=l^0\right)}{\Pr\left(\bar{L}^r(1)=\bar{l}^r \mid L^0=l^0\right)}\\
	=&\frac{\Pr \left(S^t=1\mid Z=0, L^0=l^0\right)}{\Pr \left(S^t=1 \mid Z=1, L^0=l^0\right)} \Pr \left( S^t=1 \mid \bar{L}^r=\bar{l}^r, Z=1, L^0=l^0\right).
\end{align*}
Similarly, we can obtain
	\begin{align*}
	&\Pr\left(S^{t+1}(0)=1 \mid \bar{L}^r(1)=\bar{l}^r , L^0=l^0\right)\\
	=&\frac{\Pr \left(S^{t+1}=1\mid Z=0, L^0=l^0\right)}{\Pr \left(S^{t+1}=1 \mid Z=1, L^0=l^0\right)} \Pr \left( S^{t+1}=1 \mid  \bar{L}^r=\bar{l}^r , Z=1, L^0=l^0\right).
\end{align*}
Finally, combining the results above, we have 
	\begin{align*}
	&\Pr\left(T(0)\land T(1)=t \mid \bar{L}^r(1)=\bar{l}^r, L^0=l^0\right) \\
	=&\frac{\Pr \left(S^t=1\mid Z=0, L^0=l^0\right)}{\Pr \left(S^t=1 \mid Z=1, L^0=l^0\right)} \Pr \left( S^t=1 \mid \bar{L}^r=\bar{l}^r , Z=1, L^0=l^0\right)\\
	&-\frac{\Pr \left(S^{t+1}=1\mid Z=0, L^0=l^0\right)}{\Pr \left(S^{t+1}=1 \mid Z=1, L^0=l^0\right)} \Pr \left( S^{t+1}=1 \mid \bar{L}^r=\bar{l}^r , Z=1, L^0=l^0\right).
\end{align*}

For $t=t_{\max}$, similar to above procedure, we have
\begin{align*}
&\Pr\left(T(0)\land T(1)=t_{\max} \mid \bar{L}^r(1)=\bar{l}^r, L^0=l^0\right) \\
=&\frac{\Pr \left(S^{t_{\max}}=1\mid Z=0, L^0=l^0\right)}{\Pr \left(S^{t_{\max}}=1 \mid Z=1, L^0=l^0\right)} \Pr \left( S^{t_{\max}}=1 \mid \bar{L}^r=\bar{l}^r , Z=1, L^0=l^0\right).
\end{align*}
\end{proof}

\begin{proof}[Proof of Theorem \ref{thm:identi:subA}]
	Under Assumption \ref{ass:Monoto-2}, we decompose $\mathbb{E}(Y^r| Z=1, S^r=1,\bar{L}^r,L^0)$ as follows:
	\begin{align}\label{eq:mixeddis}
		&\mathbb{E}(Y^r \mid Z=1, S^r=1, \bar{L}^r, X, A=a)\nonumber\\
		=&\mathbb{E}(Y^r \mid Z=1,G^r=LL,\bar{L}^r, X, A=a) \Pr(G^r=LL \mid S^r=1, Z=1, \bar{L}^r, X, A=a)\nonumber\\
		&+ \mathbb{E}(Y^r \mid Z=1, G^r=LD,\bar{L}^r, X, A=a) \Big( 1-\Pr(G^r=LL \mid S^r=1, Z=1, \bar{L}^r, X, A=a) \Big).
	\end{align}
	
	We have 
	\begin{align*}
		&\Pr(G^r=LL \mid S^r=1, Z=1, \bar{L}^r, X, A=a) \\
		=&   \Pr(G^r=LL \mid S^r(1)=1, \bar{L}^r(1), X, A=a)\\
		=&\Pr(S^r(0)=1 \mid S^r(1)=1, X, A=a)\\
		=&\frac{\Pr(S^r(0)=1 \mid X, A=a)}{\Pr(S^r(1)=1 \mid X, A=a)}\\=&\frac{\Pr(S^r=1 \mid Z=0,X, A=a)}{\Pr(S^r=1 \mid Z=1, X, A=a)},
	\end{align*}
	where the first equation follows Assumption \ref{ass:S-ig-2}, the second equation follows Assumptions \ref{ass:Monoto-2} and \ref{ass:surprin-ig}, and the last equation follows Assumptions \ref{ass:basic-identification}(ii)--(iii).
	Furthermore, by Assumption \ref{ass:SubRel-2}, there exist $a_0,a_1\in \mathcal{A}$ such that \(\pi_{0/1}^{r}(a_0,x)\ne\pi_{0/1}^{r}(a_1,x)\). It follows from Assumptions \ref{ass:Monoto-2} and \ref{ass:no-inter} that we have the following two equations for $a_0, a_1\in \mathcal{A}$,
	\begin{align}
		&\mathbb{E}(Y^r \mid Z=1,G^r=LL,\bar{L}^r, X, A=a_1) - \mathbb{E}(Y^r \mid Z=1,G^r=LL,\bar{L}^r, X, A=a_0)\nonumber\\
		=&\mathbb{E}(Y^r \mid Z=0,S^r=1,\bar{L}^r, X, A=a_1) - \mathbb{E}(Y^r \mid Z=0,S^r=1,\bar{L}^r, X, A=a_0),\label{eq:non-ineq1}\\
		&\mathbb{E}(Y^r \mid Z=1,G^r=LD,\bar{L}^r, X, A=a_1) - \mathbb{E}(Y^r \mid Z=1, G^r=LD, \bar{L}^r,X, A=a_0)\nonumber\\
		=&\mathbb{E}(Y^r \mid Z=0,S^r=1,\bar{L}^r, X, A=a_1) - \mathbb{E}(Y^r \mid Z=0,S^r=1,\bar{L}^r, X, A=a_0).\label{eq:non-ineq2}
	\end{align}
	Therefore, \eqref{eq:mixeddis}, \eqref{eq:non-ineq1} and \eqref{eq:non-ineq2} form a uniquely solvable system of four equations with four unknown expectations $\mathbb{E}(Y^r | Z=1, G^r=LL,\bar{L}^r, X, A=a_1)$, $\mathbb{E}(Y^r | Z=1, G^r=LL, \bar{L}^r,X, A=a_0)$, $\mathbb{E}(Y^r | Z=1,G^r=LD,\bar{L}^r, X, A=a_1)$ and $\mathbb{E}(Y^r | Z=1,G^r=LD, \bar{L}^r,X, A=a_0)$. 
	For any $a\in\mathcal{A}$, Assumptions \ref{ass:Monoto-2}
	and \ref{ass:no-inter} imply
	\begin{align*}
		&\mathbb{E}(Y^r \mid Z=1,G^r=LL,\bar{L}^r,X,A=a)\\
		=&\mathbb{E}(Y^r \mid Z=1,G^r=LL,\bar{L}^r,X,A=a_0)+\mathbb{E}(Y^r \mid Z=0,S^r=1,\bar{L}^r,X,A=a)\\
		&-\mathbb{E}(Y^r \mid Z=0,S^r=1,\bar{L}^r,X,A=a_0).
	\end{align*}
	The first term on the right-hand side has been identified above,
	and the remaining terms are observed conditional means.
	Consequently, $\mathbb{E}(Y^r \mid Z=1,G^r=LL,\bar{L}^r,X,A)$ is identifiable from the observed data distribution.

	We next derive the identification formulas in \eqref{eq:identi-0} and \eqref{eq:identi-1}. For $z=0$,
	\begin{align*}
	&\mathbb{E}\!\left[
	\mathbb{I}\{T(0)\wedge T(1)=t\}Y^r(0)
	\right] =
	\mathbb{E}\!\left[
	\mathbb{I}\{T(0)=t\}Y^r(0)
	\right] \\
	=&
	\sum_{\bar{l}^r,l^0}
	\mathbb{E}\!\left[
	Y^r(0)
	\mid
	T(0)=t,\bar{L}^r(0)=\bar{l}^r,L^0=l^0
	\right] \times
	\Pr\!\left(
	T(0)=t,\bar{L}^r(0)=\bar{l}^r
	\mid
	L^0=l^0
	\right)
	p_{L^0}(l^0) \\
	=&
	\sum_{\bar{l}^r,l^0}
	\mathbb{E}\!\left[
	Y^r
	\mid
	T=t,\bar{L}^r=\bar{l}^r,Z=0,L^0=l^0
	\right] \times
	\Pr\!\left(
	T=t,\bar{L}^r=\bar{l}^r
	\mid
	Z=0,L^0=l^0
	\right)
	p_{L^0}(l^0) \\
	=&
	\sum_{l^0}
	\mathbb{E}\!\left[
	\mathbb{I}\{T=t\}Y^r
	\mid
	Z=0,L^0=l^0
	\right]
	p_{L^0}(l^0) \\
	=&
	\mathbb{E}_{L^0}\!\left[
	\mathbb{E}\!\left\{
	\mathbb{I}\{T=t\}Y^r
	\mid
	Z=0,L^0
	\right\}
	\right] \\
	=&
	\mathbb{E}\!\left[
	\frac{\mathbb{I}\{T=t,Z=0\}}{\Pr(Z=0\mid L^0)}
	Y^r
	\right],
	\end{align*}
	where the first equation follows Assumption \ref{ass:Monoto-2}, the third equation follows Assumptions \ref{ass:basic-identification}(ii) and (iv), and the last equation follows from Assumption \ref{ass:postivity}.

	For $z=1$, it follows from Assumptions \ref{ass:Y-ig-2} and \ref{ass:noYtoS} that
    \begin{align*}
    &\mathbb{E}\!\left[
	Y^r(1)
	\mid T(0)=t,\bar{L}^r(1)=\bar{l}^r,L^0=l^0
	\right]\\
    =& \mathbb{E}\!\left[
	Y^r
	\mid T(0)=t,\bar{L}^r=\bar{l}^r, Z=1, L^0=l^0
	\right]\\
    =& \mathbb{E}\!\left[
	Y^r
	\mid G^r=LL,\bar{L}^r=\bar{l}^r, Z=1, L^0=l^0
	\right]=\mathcal{M}^r(\bar{l}^r,LL,1,a,x),
    \end{align*}
    where $l^0=(a,x^\top)^\top$. Thus, the preceding identification of $\mathcal{M}^r(\bar{L}^r,LL,1,A,X)$, together with Lemma \ref{lem::CPS}, gives
	\begin{align*}
	&\mathbb{E}\!\left[
	\mathbb{I}\{T(0)\wedge T(1)=t\}Y^r(1)
	\right]\\
	=& \mathbb{E}\!\left[
	\mathbb{E}\!\left\{Y^r(1)\mid T(0)\land T(1)=t,\bar{L}^r(1),L^0\right\}
	\Pr\left\{T(0)\land T(1)=t\mid \bar{L}^r(1),L^0\right\}
	\right]\\
	=&
	\sum_{\bar{l}^r,l^0}
	\mathbb{E}\!\left[
	Y^r(1)
	\mid
	T(0)=t,\bar{L}^r(1)=\bar{l}^r,L^0=l^0
	\right] \times
	\Pr\!\left(
	T(0)=t,\bar{L}^r(1)=\bar{l}^r
	\mid
	L^0=l^0
	\right)
	p_{L^0}(l^0) \\
	=& \sum_{\bar{l}^r,l^0} \mathcal{M}^r(\bar{l}^r,LL,1,a,x)\Big\{ \pi_{0/1}^{\,t}(a,x)
	\Pr\!\left(
	S^t=1, \bar L^r=\bar l^r\mid Z=1,L^0=l^0
	\right)\\
	&-
	\pi_{0/1}^{\,t+1}(a,x)
	\Pr\!\left(
	S^{t+1}=1, \bar L^r=\bar l^r \mid Z=1,L^0=l^0
	\right)\Big\}p_{L^0}(l^0).
	\end{align*}
	We also have 
	\begin{align*}
	&\mathbb{E}\Bigg[
	\frac{\mathbb{I}(Z=1)}{\Pr(Z=1\mid L^0)}
	\mathbb{I}(S^t=1)
	\pi_{0/1}^{\,t}(A,X)
	\mathcal{M}^r(\bar{L}^r,LL,1,A,X)
	\Bigg]\\
	=& \mathbb{E}\Bigg[
	\frac{1}{\Pr(Z=1\mid L^0)}
	\pi_{0/1}^{\,t}(A,X)
	\mathbb{E}\left\{\mathbb{I}(Z=1,S^t=1)\mathcal{M}^r(\bar{L}^r,LL,1,A,X)\mid L^0\right\}
	\Bigg]\\
	=& \mathbb{E}\Bigg[
	\pi_{0/1}^{\,t}(A,X)
	\mathbb{E}\left\{\mathbb{I}(S^t=1)\mathcal{M}^r(\bar{L}^r,LL,1,A,X)\mid Z=1, L^0\right\}
	\Bigg]\\
	=& \mathbb{E}\Bigg[
	\pi_{0/1}^{\,t}(A,X)
	\left\{\sum_{\bar{l}^r}\mathcal{M}^r(\bar{l}^r,LL,1,A,X) \Pr(S^t=1,\bar{L}^r=\bar{l}^r\mid Z=1,L^0)\right\}
	\Bigg]\\
	=&\sum_{\bar{l}^r,l^0} \mathcal{M}^r(\bar{l}^r,LL,1,a,x)\pi_{0/1}^{\,t}(a,x)\Pr\left( S^{t}=1, \bar{L}^r=\bar{l}^r \mid Z=1, L^0=l^0\right)p_{L^0}(l^0).
	\end{align*}
	Thus, the same identity holds with \(t+1\) in place of \(t\), with the term at \(t_{\max}+1\) equal to zero by convention. Finally, \eqref{eq:identi-0} and \eqref{eq:identi-1} hold. This completes the proof of Theorem \ref{thm:identi:subA}.
\end{proof}

\begin{remark}[Observed-data modeling constraints]\label{remark:obs-model-constraints}
The parameterizations in Section \ref{sec:subA-est} imply
the following modeling constraints on the observed data
distribution, which are used for parameter estimation:
\begin{align*}
&\Pr(S^r=1 \mid S^{r-1}=1, Z=1, L^0)=l_1(L^0;\beta^r);\\
&\Pr(S^r=1 \mid S^{r-1}=1, Z=0, L^0)=l_1(L^0;\beta^r)l_{0/1}(L^0;\gamma^r);\\
&\mathcal{E}^r(\bar{L}^r,0,A,X)=m(\bar{L}^r,LL,0,A,X;\alpha^r);\\
&\mathcal{E}^r(\bar{L}^r,1,A,X)= \prod_{k=1}^r l_{0/1}(L^0;\gamma^k)\, m(\bar{L}^r,LL,1,A,X;\alpha^r)
+\Big(1-\prod_{k=1}^r l_{0/1}(L^0;\gamma^k)\Big)m(\bar{L}^r,LD,1,A,X;\alpha^r).
\end{align*}
\end{remark}

\begin{remark}[Alternative parameterizations for $\mathcal{P}_2$ and $\mathcal{P}_3$]\label{remark:sur-model-restriction}
	Restricting $l_{0/1}(L^0;\gamma^r)$ in $\mathcal{P}_3$ to $(0,1)$ imposes the additional restriction
	\(
	\Pr\{S^r(0)=1\mid S^{r-1}(0)=1,L^0\}
	\le
	\Pr\{S^r(1)=1\mid S^{r-1}(1)=1,L^0\},
	\)
	for $r=1,\ldots,t_{\max}$. This restriction is not implied by constraints
	\eqref{cons:mon}--\eqref{cons:surviv}. We therefore propose an alternative parameterization that satisfies constraints \eqref{cons:mon}--\eqref{cons:surviv} without imposing this additional ordering restriction. Consider the following
	alternatives to $\mathcal{P}_2$ and $\mathcal{P}_3$, for
	$r=1,\ldots,t_{\max}$:
\begin{align*}
	\mathcal{P}_2':\quad
	\Pr\{S^r(0)=1\mid L^0\}
	&=
	\Pr\{S^{r-1}(0)=1\mid L^0\}
	l_0'(L^0;\beta^r),
	\\[4pt]
	\mathcal{P}_3':\quad
	\Pr\{S^r(1)=1\mid L^0\}
	&=
	\big\{1-l_1'(L^0;\gamma^r)\big\}
	\Pr\{S^r(0)=1\mid L^0\}
	\\
	&\quad+
	l_1'(L^0;\gamma^r)
	\Pr\{S^{r-1}(1)=1\mid L^0\},
\end{align*}
where $l_0'(\cdot;\beta^r)$ and $l_1'(\cdot;\gamma^r)$ are
regression functions taking values in $(0,1)$, and
$\beta^r$ and $\gamma^r$ are unknown finite-dimensional parameters. Then, these parameters can be estimated jointly by maximum likelihood, from which $\widehat{\pi}_{0/1}^{\,r}(A,X)$ can be obtained.
\end{remark}

\begin{proof}[Proof of Theorem \ref{thm:identi-SE}]
When \(r=0\), \(Y^0\in L^0\) is measured before treatment assignment. By Assumptions \ref{ass:Z_yPI}, \ref{ass:MTA}, \ref{ass:postivity} and \ref{ass:ranIgn},
\begin{align*}
\mathbb{E}\!\left[\mathbb{I}\{T(z_S)=t\}Y^0\right]
&= \mathbb{E}\!\left[Y^0\Pr\{T(z_S)=t\mid L^0\}\right]\\
&= \mathbb{E}\!\left[Y^0\Pr(T=t\mid Z=z_S,L^0)\right]\\
&= \mathbb{E}\left[\frac{\mathbb{I}(T=t,Z=z_S)}{\Pr(Z=z_S\mid L^0)}Y^0\right].
\end{align*}
For \(1\leq r\leq t\leq t_{\max}\), we first consider
\begin{align*}
&\mathbb{E}[\mathbb{I}(T(z_Y,z_S)=t)Y^r(z_Y, z_S)]=\mathbb{E}[Y^r(z_Y, z_S)| T(z_Y,z_S)=t]\Pr(T(z_Y,z_S)=t)\\
=& \sum_{\bar{l}^r,l^0} \mathbb{E}[Y^r(z_Y,z_S)|T(z_Y,z_S)=t,\bar{L}^r(z_Y,z_S)=\bar{l}^r, L^0=l^0]\\
& \times \Pr(T(z_Y,z_S)=t,\bar{L}^r(z_Y,z_S)=\bar{l}^r, L^0=l^0).
\end{align*}
For the term $\mathbb{E}[Y^r(z_Y,z_S)|T(z_Y,z_S)=t,\bar{L}^r(z_Y,z_S)=\bar{l}^r, L^0=l^0]$, we have
\begin{align*}
&\mathbb{E}[Y^r(z_Y,z_S)|T(z_Y,z_S)=t,\bar{L}^r(z_Y,z_S)=\bar{l}^r, L^0=l^0]\\
=&\mathbb{E}[Y^r(z_Y,z_S)|T(z_Y,z_S)=t,\bar{L}^r(z_Y,z_S)=\bar{l}^r, Z_Y(\mathcal{C})=z_Y,  Z_S(\mathcal{C})=z_S, L^0=l^0]\\
=&\mathbb{E}[Y^r(\mathcal{C})|T(\mathcal{C})=t,\bar{L}^r(\mathcal{C})=\bar{l}^r, Z_Y(\mathcal{C})=z_Y,  Z_S(\mathcal{C})=z_S, L^0=l^0]\\
=&\mathbb{E}[Y^r(\mathcal{C})|S^r(\mathcal{C})=1,\bar{L}^r(\mathcal{C})=\bar{l}^r, Z_Y(\mathcal{C})=z_Y,  Z_S(\mathcal{C})=z_Y, L^0=l^0]\\
=&\mathbb{E}[Y^r(z_Y,z_Y)|S^r(z_Y,z_Y)=1,\bar{L}^r(z_Y,z_Y)=\bar{l}^r, L^0=l^0]\\
=&\mathbb{E}[Y^r(z=z_Y)|S^r(z=z_Y)=1,\bar{L}^r(z=z_Y)=\bar{l}^r, L^0=l^0]\\
=&\mathbb{E}[Y^r|S^r=1,\bar{L}^r=\bar{l}^r, Z=z_Y , L^0=l^0],
\end{align*}
where the first equation follows from the fact that the treatment components $(Z_Y(\mathcal{C}), Z_S(\mathcal{C}))$ are randomly assigned, the third equation follows from Assumptions \ref{ass:Z_S-DCC} and \ref{ass:S-markov}, the fifth equation follows from Assumption \ref{ass:MTA}, and the last equation follows from Assumptions \ref{ass:postivity} and \ref{ass:ranIgn}.

For \(t<t_{\max}\), it follows from the law of total probability that
\begin{align*}
&\Pr(T(z_Y,z_S)=t,\bar{L}^r(z_Y,z_S)=\bar{l}^r, L^0=l^0)\\
=& \sum_{\bar{l}^{t+1}\backslash \bar{l}^{r}} \Pr(S^{t+1}(z_Y,z_S)=0|S^{t}(z_Y,z_S)=1, \bar{L}^{t+1}(z_Y,z_S)=\bar{l}^{t+1}, L^0=l^0)\\
&\times\Pr({L}^{t+1}(z_Y,z_S)={l}^{t+1}|S^{t}(z_Y,z_S)=1, \bar{L}^{t}(z_Y,z_S)=\bar{l}^{t}, L^0=l^0)\\
&\times \prod_{i=1}^t \Big\{  \Pr(S^{i}(z_Y,z_S)=1|S^{i-1}(z_Y,z_S)=1, \bar{L}^{i}(z_Y,z_S)=\bar{l}^{i}, L^0=l^0)\\
&\times\Pr({L}^{i}(z_Y,z_S)={l}^{i}|S^{i-1}(z_Y,z_S)=1, \bar{L}^{i-1}(z_Y,z_S)=\bar{l}^{i-1}, L^0=l^0)\Big\}\times p_{L^0}(l^0).
\end{align*}
Then, we have 
\begin{align*}
&\Pr(S^{t+1}(z_Y,z_S)=0|S^{t}(z_Y,z_S)=1, \bar{L}^{t+1}(z_Y,z_S)=\bar{l}^{t+1}, L^0=l^0)\\
=&\Pr(S^{t+1}(z_Y,z_S)=0|S^{t}(z_Y,z_S)=1, \bar{L}^{t+1}(z_Y,z_S)=\bar{l}^{t+1}, Z_Y(\mathcal{C})=z_Y,  Z_S(\mathcal{C})=z_S, L^0=l^0)\\
=&\Pr(S^{t+1}(\mathcal{C})=0|S^{t}(\mathcal{C})=1, \bar{L}^{t+1}(\mathcal{C})=\bar{l}^{t+1}, Z_Y(\mathcal{C})=z_Y,  Z_S(\mathcal{C})=z_S, L^0=l^0)\\
=&\Pr(S^{t+1}(\mathcal{C})=0|S^{t}(\mathcal{C})=1, \bar{L}^{t+1}(\mathcal{C})=\bar{l}^{t+1}, Z_Y(\mathcal{C})=z_S,  Z_S(\mathcal{C})=z_S, L^0=l^0)\\
=&\Pr(S^{t+1}(z_S,z_S)=0|S^{t}(z_S,z_S)=1, \bar{L}^{t+1}(z_S,z_S)=\bar{l}^{t+1}, L^0=l^0)\\
=&\Pr(S^{t+1}(z=z_S)=0|S^{t}(z=z_S)=1, \bar{L}^{t+1}(z=z_S)=\bar{l}^{t+1}, L^0=l^0)\\
=&\Pr(S^{t+1}=0|S^{t}=1, \bar{L}^{t+1}=\bar{l}^{t+1}, Z=z_S, L^0=l^0),
\end{align*}
where the first equality follows from the fact that the treatment components $(Z_Y(\mathcal{C}), Z_S(\mathcal{C}))$ are randomly assigned, the third equation follows from Assumption \ref{ass:DCC_LongSE-2}(1), the fifth equation follows from Assumption \ref{ass:MTA}, and the last equation follows from Assumption \ref{ass:ranIgn}.
Similarly, it follows from Assumptions \ref{ass:MTA}, \ref{ass:ranIgn} and \ref{ass:DCC_LongSE-2}(2) that we have 
\begin{align*}
&\Pr({L}^{t+1}(z_Y,z_S)={l}^{t+1}|S^{t}(z_Y,z_S)=1, \bar{L}^{t}(z_Y,z_S)=\bar{l}^{t}, L^0=l^0)\\
=&\Pr({L}^{t+1}={l}^{t+1}|S^{t}=1, \bar{L}^{t}=\bar{l}^{t}, Z=z_S, L^0=l^0).
\end{align*}
Furthermore, according to Assumptions \ref{ass:MTA}, \ref{ass:DCC_LongSE-2} and \ref{ass:ranIgn}, we repeat the above procedures and obtain
\begin{align*}
&\Pr(S^{i}(z_Y,z_S)=1|S^{i-1}(z_Y,z_S)=1, \bar{L}^{i}(z_Y,z_S)=\bar{l}^{i}, L^0=l^0)\\
=&\Pr(S^{i}=1|S^{i-1}=1, \bar{L}^{i}=\bar{l}^{i}, Z=z_S, L^0=l^0),\\
&\Pr({L}^{i}(z_Y,z_S)={l}^{i}|S^{i-1}(z_Y,z_S)=1, \bar{L}^{i-1}(z_Y,z_S)=\bar{l}^{i-1}, L^0=l^0)\\
=&\Pr({L}^{i}={l}^{i}|S^{i-1}=1, \bar{L}^{i-1}=\bar{l}^{i-1}, Z=z_S, L^0=l^0).
\end{align*}
Combining the above equalities, we obtain
\begin{align*}
&\Pr(T(z_Y,z_S)=t,\bar{L}^r(z_Y,z_S)=\bar{l}^r, L^0=l^0)\\
=& \sum_{\bar{l}^{t+1}\backslash \bar{l}^{r}} \Pr(S^{t+1}=0|S^{t}=1, \bar{L}^{t+1}=\bar{l}^{t+1}, Z=z_S,L^0=l^0)\\
&\times\Pr({L}^{t+1}={l}^{t+1}|S^{t}=1, \bar{L}^{t}=\bar{l}^{t}, Z=z_S, L^0=l^0)\\
&\times \prod_{i=1}^t \Big\{  \Pr(S^{i}=1|S^{i-1}=1, \bar{L}^{i}=\bar{l}^{i}, Z=z_S, L^0=l^0)\\
&\times\Pr({L}^{i}={l}^{i}|S^{i-1}=1, \bar{L}^{i-1}=\bar{l}^{i-1}, Z=z_S, L^0=l^0)\Big\}\times p_{L^0}(l^0)\\
=&\Pr(T=t,\bar{L}^r=\bar{l}^r|Z=z_S,L^0=l^0)p_{L^0}(l^0).
\end{align*}
Similarly, for \(t=t_{\max}\),
\[
\begin{aligned}
	\Pr\left(
	T(z_Y,z_S)=t_{\max},
	\bar L^r(z_Y,z_S)=\bar l^r,
	L^0=l^0
	\right)\
	=
	\Pr\left(
	T=t_{\max},\bar L^r=\bar l^r
	\mid Z=z_S,L^0=l^0
	\right)p_{L^0}(l^0).
\end{aligned}
\]

Finally, combining all the above results, we have
\begin{align*}
&\mathbb{E}[\mathbb{I}(T(z_Y,z_S)=t)Y^r(z_Y, z_S)]\\
=&\sum_{\bar{l}^r,l^0} \mathbb{E}(Y^r | S^r=1, \bar{L}^r=\bar{l}^r, Z=z_Y, L^0=l^0)\Pr(T=t, \bar{L}^r=\bar{l}^r | Z=z_S, L^0=l^0)p_{L^0}(l^0)\\
=& \mathbb{E}\Big \{ \mathbb{E}\Big[ \mathbb{E}(Y^r | S^r=1, \bar{L}^r, Z=z_Y, L^0)\mathbb{I}(T=t)\Big | Z=z_S, L^0\Big] \Big\}.
\end{align*}
This completes the proof of Theorem \ref{thm:identi-SE}.
\end{proof}

\begin{proof}[Proof of Proposition \ref{prop:connection}]
For $z=0,1$ and $r\leq t$, it follows from Assumption \ref{ass:MTA} that we have
\begin{align*}
\mathbb{E}[\mathbb{I}(T(z)=t)Y^r(z)]=
\mathbb{E}[\mathbb{I}(T(z_Y=z,z_S=z)=t)Y^r(z_Y=z,z_S=z)],
\end{align*}
which implies $\mu(0)=\Gamma(0,0)$ if $w_t^r=\psi_t^r$.
Then, consider
\begin{align*}
&\mathbb{E}[\mathbb{I}(T(0)\land T(1)=t)Y^r(1)] \\
=&\mathbb{E}[\mathbb{I}(T(0)=t)Y^r(1)]\\
=& \mathbb{E}[\mathbb{I}(T(z_Y=0,z_S=0)=t)Y^r(z_Y=1, z_S=1)]\\
=& \mathbb{E}[\mathbb{I}(T(z_Y=0,z_S=0)=t)Y^r(z_Y=1, z_S=0)]\\
=& \mathbb{E}[\mathbb{I}(T(z_S=0)=t)Y^r(z_Y=1, z_S=0)],
\end{align*}
where the first, second, third, and last equalities follow from Assumptions \ref{ass:Monoto-2}, \ref{ass:MTA}, \ref{ass:Z_sPI}, and \ref{ass:Z_yPI}, respectively, and thus there holds true that $\mu(1)=\Gamma(1,0)$ if $w_t^r=\psi_t^r$. Finally, when $\phi_t^r=\psi_t^r$, we also have $\mu^{\text{ext}}=\Gamma(1,1)-\Gamma(1,0)$. This completes the proof of Proposition \ref{prop:connection}.
 
\end{proof}

We characterize the nonparametric influence function $\textbf{IF}(\Lambda_{t}^r(z_Y,z_S))$ of $\Lambda_t^r(z_Y,z_S)$ for $r \ge 1$ without imposing parametric restrictions on the observed data distribution \citep{bickel1993efficient, tsiatis2006semiparametric} in the following Theorem.

\begin{theorem}\label{thm:triplyRobust-SE}
	The nonparametric influence function of the estimand \eqref{eq:iden-SE} is given by
	\begin{align*}
		\textbf{IF}(\Lambda_{t}^r(z_Y,z_S)) 
		=& \mathbb{E}\left\{ \mathbb{I}(T=t)\mathbb{E}\left(Y^r \mid S^r=1, \bar{L}^r, Z=z_Y, L^0\right)  \Big | Z=z_S, L^0\right\} \\
		& + \frac{\mathbb{I}( Z=z_S)}{\Pr( Z=z_S\mid L^0)}\Big [\mathbb{I}(T=t)\mathbb{E}(Y^r \mid S^r=1, \bar{L}^r, Z=z_Y, L^0) \\
		&\qquad - \mathbb{E}\left\{\mathbb{I}(T=t)\mathbb{E}(Y^r \mid S^r=1, \bar{L}^r, Z=z_Y, L^0) \mid Z=z_S, L^0\right\} \Big]\\
		&+ \frac{\mathbb{I}(Z=z_Y)}{\Pr(Z=z_Y\mid L^0)} 
		\frac{p_{T,\bar{L}^r \mid Z,L^0}(t, \bar{L}^r \mid z_S,L^0)}{p_{S^r,\bar{L}^r \mid Z,L^0}(1,\bar{L}^r \mid z_Y,L^0)}
		S^r\left\{Y^r-\mathbb{E}(Y^r\mid S^r=1, \bar{L}^r, Z=z_Y,L^0)\right\}\\
		&- \Lambda_t^r(z_Y,z_S).
	\end{align*}
\end{theorem}

\begin{proof}[Proof of Theorem \ref{thm:triplyRobust-SE}]
Consider a regular parametric submodel $\{p_\epsilon\}$
through the observed-data distribution $p$ at $\epsilon=0$.
Let $\mathbb{E}_\epsilon$, $\Pr_\epsilon$, and $f_\epsilon$
denote expectation, probability, and the relevant density
or mass function under $p_\epsilon$. Write
\(
S(\cdot)=
\left.\frac{\partial}{\partial\epsilon}
\log f_\epsilon(\cdot)\right|_{\epsilon=0},
\)
with analogous notation for conditional scores. Let
\begin{align*}
\nu(p)=&\mathbb{E}[\mathbb{I}(t=T(z_S))Y^r(z_Y, z_S)]\\
=& \mathbb{E}\Big \{ \mathbb{E}\Big[ \mathbb{E}(Y^r | S^r=1, \bar{L}^r, Z=z_Y, L^0)\mathbb{I}(T=t)\Big | Z=z_S, L^0\Big] \Big\}\\
=& \mathbb{E}\Big \{ \mathbb{E}\Big[ \mathbb{E}(Y^r | S^r=1, \bar{L}^r, Z=z_Y, L^0)\Big | T=t, Z=z_S, L^0\Big] \Pr(T=t| Z=z_S, L^0) \Big\},
\end{align*}
and consider 
\begin{align*}
&\frac{d \nu(p_{\epsilon})}{d \epsilon} \Big|_{\epsilon=0}\\
=& \int \mathbb{E}_\epsilon\Big[ \mathbb{E}_\epsilon(Y^r | S^r=1, \bar{L}^r, Z=z_Y, L^0)\Big | T=t, Z=z_S, L^0\Big] {\Pr}_{\epsilon}(T=t| Z=z_S, L^0) \frac{d f_\epsilon(l^0)}{d \epsilon} dl^0 \Big|_{\epsilon=0}\\
&+ \mathbb{E}_\epsilon\Big \{ \int \mathbb{E}_\epsilon(Y^r | S^r=1, \bar{L}^r, Z=z_Y, L^0) \frac{df_\epsilon(\bar{l}^r| T=t, Z=z_S, L^0 )}{d\epsilon}d\bar{l}^r {\Pr}_{\epsilon}(T=t| Z=z_S, L^0) \Big\} \Big|_{\epsilon=0}\\
&+  \mathbb{E}_\epsilon\Big \{ \mathbb{E}_\epsilon \Big[ \int y^r \frac{df_\epsilon(y^r | S^r=1, \bar{L}^r, Z=z_Y, L^0)}{d\epsilon}dy^r \Big | T=t, Z=z_S, L^0\Big] {\Pr}_{\epsilon}(T=t| Z=z_S, L^0) \Big\}\Big|_{\epsilon=0}\\
&+ \mathbb{E}_\epsilon \Big \{ \mathbb{E}_\epsilon\Big[ \mathbb{E}_\epsilon(Y^r | S^r=1, \bar{L}^r, Z=z_Y, L^0)\Big | T=t, Z=z_S, L^0\Big] \frac{df_\epsilon(T=t| Z=z_S, L^0)}{d\epsilon} \Big\} \Big|_{\epsilon=0}\\
:=& T_1 + T_2 + T_3 + T_4.
\end{align*}
For $T_1$, $T_2$, $T_3$ and $T_4$, we derive
\begin{align*}
T_1 =& \mathbb{E}\Big \{ \mathbb{E}\Big[ \mathbb{E}(Y^r | S^r=1, \bar{L}^r, Z=z_Y, L^0)\Big | T=t, Z=z_S, L^0\Big] \Pr(T=t| Z=z_S, L^0) S(L^0)\Big\}\\
=& \mathbb{E}\Big \{ \mathbb{E}\Big[ \mathbb{E}(Y^r | S^r=1, \bar{L}^r, Z=z_Y, L^0)\Big | T=t, Z=z_S, L^0\Big] \Pr(T=t| Z=z_S, L^0) S(Y^r, T, S^r, \bar{L}^r, Z, L^0)\Big\},
\end{align*}

\begin{align*}
&T_2 \\
= &\mathbb{E}\Big \{ \mathbb{E}\Big[ \mathbb{E}(Y^r | S^r=1, \bar{L}^r, Z=z_Y, L^0) S(\bar{L}^r|T=t, Z=z_S,L^0)\Big | T=t, Z=z_S, L^0\Big] \Pr(T=t| Z=z_S, L^0) \Big\}\\
=& \mathbb{E}\Big \{ \mathbb{E}\Big[ \frac{\mathbb{I}(T=t, Z=z_S)}{\Pr(T=t, Z=z_S|L^0)}\mathbb{E}(Y^r | S^r=1, \bar{L}^r, Z=z_Y, L^0) S(\bar{L}^r|T, Z,L^0)\Big | L^0\Big] \Pr(T=t| Z=z_S, L^0) \Big\}\\
=& \mathbb{E}\Big \{ \mathbb{E}\Big[ \frac{\mathbb{I}(T=t, Z=z_S)}{\Pr( Z=z_S|L^0)}\mathbb{E}(Y^r | S^r=1, \bar{L}^r, Z=z_Y, L^0) S(\bar{L}^r|T, Z,L^0)\Big | L^0\Big]  \Big\}\\
=& \mathbb{E}\Big \{\frac{\mathbb{I}(T=t, Z=z_S)}{\Pr( Z=z_S|L^0)}\mathbb{E}(Y^r | S^r=1, \bar{L}^r, Z=z_Y, L^0) S(\bar{L}^r|T, Z,L^0) \Big\}\\
=& \mathbb{E}\Big \{\frac{\mathbb{I}(T=t, Z=z_S)}{\Pr( Z=z_S|L^0)}\Big (\mathbb{E}(Y^r | S^r=1, \bar{L}^r, Z=z_Y, L^0) \\
&- \mathbb{E}[\mathbb{E}(Y^r | S^r=1, \bar{L}^r, Z=z_Y, L^0) | T=t, Z=z_S, L^0] \Big ) S(\bar{L}^r|T, Z,L^0) \Big\}\\
=& \mathbb{E}\Big \{\frac{\mathbb{I}(T=t, Z=z_S)}{\Pr( Z=z_S|L^0)}\Big (\mathbb{E}(Y^r | S^r=1, \bar{L}^r, Z=z_Y, L^0) \\
&- \mathbb{E}[\mathbb{E}(Y^r | S^r=1, \bar{L}^r, Z=z_Y, L^0) | T=t, Z=z_S, L^0] \Big ) S(\bar{L}^r,T, Z,L^0) \Big\}\\
=& \mathbb{E}\Big \{\frac{\mathbb{I}(T=t, Z=z_S)}{\Pr( Z=z_S|L^0)}\Big (\mathbb{E}(Y^r | S^r=1, \bar{L}^r, Z=z_Y, L^0) \\
&- \mathbb{E}[\mathbb{E}(Y^r | S^r=1, \bar{L}^r, Z=z_Y, L^0) | T=t, Z=z_S, L^0] \Big ) S(Y^r,T,S^r, \bar{L}^r, Z, L^0) \Big\},
\end{align*}

\begin{align*}
&T_3 \\
= & \mathbb{E}\Big \{ \mathbb{E}\Big[ \mathbb{E}(Y^r S(Y^r|S^r=1, \bar{L}^r, Z=z_Y, L^0) | S^r=1, \bar{L}^r, Z=z_Y, L^0)\Big | T=t, Z=z_S, L^0\Big] \\
& \times \Pr(T=t| Z=z_S, L^0) \Big\}\\
= &  \mathbb{E}\Big \{ \mathbb{E}\Big[ \frac{\mathbb{I}(T=t, Z=z_S)}{\Pr(Z=z_S | L^0)} \mathbb{E}(Y^r S(Y^r|S^r=1, \bar{L}^r, Z=z_Y, L^0) | S^r=1, \bar{L}^r, Z=z_Y, L^0)\Big |  L^0\Big]\Big \}\\
= &  \mathbb{E}\Big \{  \frac{\mathbb{I}(T=t, Z=z_S)}{\Pr(Z=z_S | L^0)} \mathbb{E}(\frac{\mathbb{I}(S^r=1,Z=z_Y)}{\Pr(S^r=1,Z=z_Y| \bar{L}^r, L^0)}Y^r S(Y^r|S^r, \bar{L}^r, Z, L^0) |  \bar{L}^r, L^0)\Big \}\\
= & \mathbb{E}\Big \{  \frac{\Pr(T=t, Z=z_S|\bar{L}^r, L^0)}{\Pr(Z=z_S | L^0)} \frac{\mathbb{I}(Z=z_Y)S^r}{\Pr(S^r=1,Z=z_Y| \bar{L}^r, L^0)} Y^r S(Y^r|S^r, \bar{L}^r, Z, L^0) \Big \}\\
=& \mathbb{E}\Big \{  \frac{\Pr(T=t, Z=z_S|\bar{L}^r, L^0)}{\Pr(Z=z_S | L^0)} \frac{\mathbb{I}(Z=z_Y)S^r}{\Pr(S^r=1,Z=z_Y| \bar{L}^r, L^0)} \\
& \times (Y^r-\mathbb{E}(Y^r|S^r=1, \bar{L}^r,Z=z_Y,L^0)) S(Y^r|S^r, \bar{L}^r, Z, L^0) \Big \}\\
=& \mathbb{E}\Big \{  \frac{\Pr(T=t, Z=z_S|\bar{L}^r, L^0)}{\Pr(Z=z_S | L^0)} \frac{\mathbb{I}(Z=z_Y)S^r}{\Pr(S^r=1,Z=z_Y| \bar{L}^r, L^0)} \\
& \times (Y^r-\mathbb{E}(Y^r|S^r=1, \bar{L}^r,Z=z_Y,L^0)) S(Y^r,T,S^r, \bar{L}^r, Z, L^0) \Big \},
\end{align*}
and 
\begin{align*}
&T_4\\
=& \mathbb{E}\Big \{ \mathbb{E}\Big[ \mathbb{E}(Y^r | S^r=1, \bar{L}^r, Z=z_Y, L^0)\Big | T=t, Z=z_S, L^0\Big] \mathbb{E}(\mathbb{I}(T=t)S(T|Z=z_S,L^0)| Z=z_S, L^0) \Big\}\\
=& \mathbb{E}\Big \{ \mathbb{E}\Big[ \mathbb{E}(Y^r | S^r=1, \bar{L}^r, Z=z_Y, L^0)\Big | T=t, Z=z_S, L^0\Big] \mathbb{E}(\frac{\mathbb{I}(Z=z_S)}{\Pr(Z=z_S|L^0)}\mathbb{I}(T=t)S(T|Z,L^0)| L^0) \Big\}\\
=& \mathbb{E}\Big \{ \mathbb{E}\Big[ \mathbb{E}(Y^r | S^r=1, \bar{L}^r, Z=z_Y, L^0)\Big | T=t, Z=z_S, L^0\Big] \frac{\mathbb{I}(Z=z_S)}{\Pr(Z=z_S|L^0)}\\
&\times (\mathbb{I}(T=t)-\Pr(T=t|Z=z_S,L^0))S(T|Z,L^0)\Big\}\\
=&\mathbb{E}\Big \{ \mathbb{E}\Big[ \mathbb{E}(Y^r | S^r=1, \bar{L}^r, Z=z_Y, L^0)\Big | T=t, Z=z_S, L^0\Big] \frac{\mathbb{I}(Z=z_S)}{\Pr(Z=z_S|L^0)}\\
&\times (\mathbb{I}(T=t)-\Pr(T=t|Z=z_S,L^0))S(T,Z,L^0)\Big\}\\
=&\mathbb{E}\Big \{ \mathbb{E}\Big[ \mathbb{E}(Y^r | S^r=1, \bar{L}^r, Z=z_Y, L^0)\Big | T=t, Z=z_S, L^0\Big] \frac{\mathbb{I}(Z=z_S)}{\Pr(Z=z_S|L^0)}\\
&\times (\mathbb{I}(T=t)-\Pr(T=t|Z=z_S,L^0))S(Y^r,T,S^r, \bar{L}^r, Z, L^0)\Big\}.
\end{align*}

Thus, the influence function of $\nu(p)$ is 
\begin{align*}
\nu^1(p) =& \mathbb{E}\Big[ \mathbb{E}(Y^r | S^r=1, \bar{L}^r, Z=z_Y, L^0)\Big | T=t, Z=z_S, L^0\Big] \Pr(T=t| Z=z_S, L^0) - \nu(p) \\
& + \frac{\mathbb{I}( T=t, Z=z_S)}{\Pr( Z=z_S|L^0)}\Big (\mathbb{E}(Y^r | S^r=1, \bar{L}^r, Z=z_Y, L^0) \\
&~ - \mathbb{E}[\mathbb{E}(Y^r | S^r=1, \bar{L}^r, Z=z_Y, L^0) | T=t, Z=z_S, L^0] \Big )\\
&+   \frac{\Pr(T=t, Z=z_S|\bar{L}^r, L^0)}{\Pr(Z=z_S | L^0)} \frac{\mathbb{I}(Z=z_Y)S^r}{\Pr(S^r=1,Z=z_Y| \bar{L}^r, L^0)} (Y^r-\mathbb{E}(Y^r|S^r=1, \bar{L}^r,Z=z_Y,L^0))\\
&+ \mathbb{E}\Big[ \mathbb{E}(Y^r | S^r=1, \bar{L}^r, Z=z_Y, L^0)\Big | T=t, Z=z_S, L^0\Big] \frac{\mathbb{I}(Z=z_S)}{\Pr(Z=z_S|L^0)}\\
&~\times (\mathbb{I}(T=t)-\Pr(T=t|Z=z_S,L^0))\\
=& \mathbb{E}\Big[ \mathbb{I}(T=t)\mathbb{E}(Y^r | S^r=1, \bar{L}^r, Z=z_Y, L^0)  \Big | Z=z_S, L^0\Big] - \nu(p) \\
& + \frac{\mathbb{I}( Z=z_S)}{\Pr( Z=z_S|L^0)}\Big (\mathbb{I}(T=t)\mathbb{E}(Y^r | S^r=1, \bar{L}^r, Z=z_Y, L^0) \\
&~ - \mathbb{E}[\mathbb{I}(T=t)\mathbb{E}(Y^r | S^r=1, \bar{L}^r, Z=z_Y, L^0) | Z=z_S, L^0] \Big )\\
&+    \frac{\mathbb{I}(Z=z_Y)}{\Pr(Z=z_Y|L^0)} \frac{p_{T,\bar{L}^r \mid Z,L^0}(t, \bar{L}^r \mid z_S,L^0)}{p_{S^r,\bar{L}^r \mid Z,L^0}(1,\bar{L}^r \mid z_Y,L^0)}S^r(Y^r-\mathbb{E}(Y^r|S^r=1, \bar{L}^r,Z=z_Y,L^0)).
\end{align*}

\end{proof}

\suppsection{Identification and estimation of while-alive estimands}{supple:while-alive}

The following Table \ref{tab:example-weight-while-alive} provides four example types of weighting schemes for the while-alive estimand $\lambda(z)$.

\begin{table}[!htbp]
  \centering
  \caption{Four example types of while-alive estimands with their corresponding weighting schemes}
  \label{tab:example-weight-while-alive} 
  \renewcommand{\arraystretch}{1.2}
  \begin{tabular}{lll}
    \toprule
    Type &  Weight \\
    \midrule
    Exit time &  $\zeta_{T(z)}^r=
\begin{cases}
1, & r=T(z),\\
0,  & r<T(z).
\end{cases}$ \\
    Average &  $\zeta_{T(z)}^r=\frac{1}{T(z)+1}$ \\
    Cumulative &  $\zeta_{T(z)}^r=1$ \\
    AUC-based &  $\zeta_{T(z)}^r= \mathbb{I}(T(z)\geq 1) \times
\begin{cases}
\frac{1}{2}(\tau_{T(z)} - \tau_{{T(z)}-1}), & r=T(z),\\
\frac{1}{2}(\tau_{r+1}-\tau_{r-1}),  & 0<r< {T(z)},\\
\frac{1}{2}\tau_1, & r=0.
\end{cases}$ \\
    \bottomrule
  \end{tabular}
\renewcommand{\arraystretch}{1}
\end{table}

Under Assumptions \ref{ass:basic-identification}(ii)--(iv), $\lambda(z)$ is identifiable from the observed distribution. Then, the while-alive estimator for $\lambda(z)$ with $z=0,1$ is
\begin{align*}
\widehat{\lambda}(z)=\sum_{t=0}^{t_{\max}} \sum_{r=0}^t \zeta_t^r \mathbb{P}_n \Big \{\frac{\mathbb{I}(T=t, Z=z)}{zl_z(L^0;\widehat{\theta})+(1-z)(1-l_z(L^0;\widehat{\theta}))}Y^r \Big \}.
\end{align*}

\suppsection{Simulation}{supple:simulation}

We examine the numerical performance of the proposed estimators in Sections \ref{sec:subA-est} and \ref{sec:estimator-SE}, respectively. The estimators are evaluated in terms of mean bias and its Monte Carlo standard error, both scaled by ($10^3$).

\subsection*{Substitution variable approach}

The specific data-generating mechanism is described as follows:
\begin{enumerate}[label=(\arabic*)]
\item Generate baseline covariates $X=(X_1,X_2,X_3)^\top$: $X_1$ is a discrete variable taking values $-1$ and $1$ with probabilities $\Pr(X_1=-1)=\Pr(X_1=1)=0.5$; 
\(
X_2,X_3\stackrel{\text{ind}}{\sim}\mathrm{Uniform}(-1,1);
\)

\item Generate substitution variable $A$ from a Bernoulli distribution with $\Pr(A=1 | X)= \text{expit}(\iota^\top X)$, where $\iota=(0.2,0.1,-0.1)^\top$;

\item Generate exposure variable $Z$ from a Bernoulli distribution with $\Pr(Z=1| L^0)=\text{expit}(\theta_z^\top L^0)$, where $L^0=(A, X^\top)^\top$ and $\theta_z=(0.1,0.2,-0.1,0.1)^\top$;

\item Set $t_{\max}=3$, $\tau=(0,1/4,1/2,1)$, $S^0=1$ and $Y^0=0$;

\item Generate the latent death, harmful, and progression processes as follows. For \(r=1,2,3\), let \(p_r(L^0)\), \(\rho_r(L^0)\), and \(s_r(L^0)\) be the conditional probabilities that these three processes, respectively, do not occur during interval \(r\), given that they have not occurred before interval \(r\). Specifically,
\[
\begin{aligned}
	p_1(L^0)&=\operatorname{expit}(2.2+0.2A+0.3X_1-0.2X_2+0.1X_3),\\
	p_2(L^0)&=\operatorname{expit}(2.1+0.2A+0.3X_1-0.2X_2+0.1X_3),\\
	p_3(L^0)&=\operatorname{expit}(2.0+0.2A+0.3X_1-0.2X_2+0.1X_3),\\
	\rho_1(L^0)
	&=\operatorname{expit}(0.9+1.5A-0.2X_1+0.1X_2+0.2X_3),\\
	\rho_2(L^0)
	&=\operatorname{expit}(0.7+1.5A-0.2X_1+0.1X_2+0.2X_3),\\
	\rho_3(L^0)
	&=\operatorname{expit}(0.6+1.5A-0.2X_1+0.1X_2+0.2X_3),\\
	s_1(L^0)&=\operatorname{expit}(0.5-0.2A+0.1X_1-0.2X_2+0.1X_3),\\
	s_2(L^0)&=\operatorname{expit}(0.4-0.2A+0.1X_1-0.2X_2+0.1X_3),\\
	s_3(L^0)&=\operatorname{expit}(0.3-0.2A+0.1X_1-0.2X_2+0.1X_3).
\end{aligned}
\]
First generate nine independent \(\operatorname{Uniform}(0,1)\) variables \(U_D^r,U_H^r,U_P^r\), \(r=1,2,3\), independently of \(Z\). For \(r=2,3\), independently replace \(U_D^r\) by \(U_P^{r-1}\) with probability \(0.5\), and otherwise preserve its original value. Define
\[
B_D^r=\mathbb I\{U_D^r\le p_r(L^0)\},\qquad
B_H^r=\mathbb I\{U_H^r\le\rho_r(L^0)\},\qquad
B_P^r=\mathbb I\{U_P^r\le s_r(L^0)\}.
\]
Set \(R_D^0=R_H^0=R_P^0=1\). Define 
\(
R_D^r=R_D^{r-1}B_D^r,
R_H^r=R_H^{r-1}B_H^r,\) and \(
R_P^r=R_P^{r-1}B_P^r, 
\) recursively.

\item Construct \(S^r(z)\), \(L^r(z)\), \(T(z)\), and the survival type \(G^r\). For \(r=1,2,3\), define
\[
S^r(0)=R_H^rR_D^r,\qquad 
S^r(1)=R_D^r,
\]
and
\[
L^r(0)=1-R_H^rR_P^rR_D^r,\qquad
L^r(1)=1-R_P^rR_D^r.
\]
For \(z=0,1\), these expressions for \(L^r(z)\) apply only when \(S^{r-1}(z)=1\); otherwise, \(L^r(z)\) is undefined. Then let
\(
T(z)=\sum_{r=1}^3 S^r(z),
\)
and define \(G^r\) from the pair \((S^r(1),S^r(0))\).

\item Generate $Y^r(0)$ if $S^r(0)=1$ and $Y^r(1)$ if $S^r(1)=1$, for $r=1,2,3$. Let
\(
B(L^0)=0.5+0.2A+0.3X_1-0.2X_2+0.2X_3
\), $\eta_Z=0.5$ and \(\eta_G=0.1\). (Code $LL$ to be $1$ and $LD$ to be $0$)
\begin{itemize}
\item Generate $Y^1(1)$ from a normal distribution $N(\xi^1_1,0.5^2)$, where $\xi_1^1=B(L^0)+L^1(1)+\eta_Z+\eta_GG^1$, if $S^1(1)=1$; generate $Y^r(1)$  from a normal distribution $N(\xi^r_1,0.5^2)$, where $\xi_1^r=B(L^0)+L^r(1)+ \eta_Z+r\eta_GG^r + Y^{r-1}(1)-(r-1)\eta_G G^{r-1}$, if $S^r(1)=1$ for $r=2,3$;
\item Generate $Y^1(0)$ from a normal distribution $N(\xi^1_0,0.5^2)$, where $\xi_0^1=B(L^0)+L^1(0)$, if $S^1(0)=1$; generate $Y^r(0)$  from a normal distribution $N(\xi^r_0,0.5^2)$, where $\xi_0^r=B(L^0)+L^r(0)+ Y^{r-1}(0)$, if $S^r(0)=1$ for $r=2,3$;
\end{itemize}

\item Reveal the observed survival status
\(
S^r=ZS^r(1)+(1-Z)S^r(0).
\)
If \(S^{r-1}=1\), 
\(
L^r=ZL^r(1)+(1-Z)L^r(0);\)
if \(S^r=1\), 
\(
Y^r=ZY^r(1)+(1-Z)Y^r(0).
\)
The covariate \(L^r\) is undefined when \(S^{r-1}=0\), whereas \(Y^r\) is undefined when \(S^r=0\). Finally, \(T=\sum_{r=1}^3S^r\).
\end{enumerate}

We are interested in estimating $\mu(1)-\mu(0)$ and $\mu^{\text{ext}}$ under different weighting schemes, and the corresponding true values are
presented in Table \ref{tab:trueValue_SubA}. Table \ref{tab:re_simul_SubA} summarizes the simulation results based on sample sizes of $n=500,2000$ with $500$ independent replications.

\begin{table}[htbp]
	\centering
	\caption{The true values of $\mu(1)-\mu(0)$ and $\mu^{\text{ext}}$ under different weighting schemes}
	\label{tab:trueValue_SubA}
	\begin{tabular}{lcc}
		\toprule
		weight & $\mu(1)-\mu(0)$  & $\mu^{\text{ext}}$ \\
		\midrule
			Exit time & 0.99  &  -  \\
			Average & 0.49  & -  \\
			Cumulative & 1.76   & 2.87 \\
			AUC-based &  0.46  & 0.90  \\
		\bottomrule
	\end{tabular}
\end{table}

\begin{table}[!htb]
	\centering
	\caption{Bias ($\times 10^3$) and its Monte Carlo standard error ($\times 10^3$, shown in parentheses) for estimating $\mu(1)-\mu(0)$ and $\mu^{\text{ext}}$ with different choices of weights; results are based on sample sizes of $n=500,2000$ with $500$ independent replications}
	\label{tab:re_simul_SubA}
	\begin{tabular}{llcc}
		\toprule
		Sample size & Weight  & $\mu(1)-\mu(0)$ & $\mu^{\text{ext}}$\\
		\midrule
		\multirow{4}{*}{$n=500$} 
		& Exit time & $-$12.83 (6.92) & - \\
		& Average & $-$5.30 (3.42) & - \\
		& Cumulative & $-$17.49 (12.58) & 22.65 (18.48) \\
		& AUC-based & $-$4.27 (3.44) & 5.02 (5.58) \\
		\midrule
		\multirow{4}{*}{$n=2000$} 
		& Exit time & $-$3.48 (3.60) & - \\
		& Average & $-$1.87 (1.81) & - \\
		& Cumulative & $-$6.78 (6.55) & $-$2.25 (9.54)\\
		& AUC-based & $-$1.69 (1.77) & $-$0.93 (2.87)\\
		\bottomrule
	\end{tabular}
	
\end{table}

\subsection*{Conditional separable effect approach}
We simulate samples according to Figure \ref{fig:Longi-SE-2} and the specific data-generating mechanism is described as follows:
\begin{enumerate}[label=(\arabic*)]
	\item Generate baseline covariates $L^0=(A,X^\top)^\top$;
	\begin{itemize}
		\item $X=(X_1,X_2,X_3)^\top$ comprises discrete and continuous variables: $X_1$ is a discrete variable taking values $-1$ and $1$ with probabilities $\Pr(X_1=-1)=\Pr(X_1=1)=0.5$; 
		\(
		X_2,X_3\stackrel{\text{ind}}{\sim}\mathrm{Uniform}(-1,1);
		\)
		\item $A$ is generated from a Bernoulli distribution with $\Pr(A=1 | X)= \text{expit}(\iota^\top X)$, where $\iota=(0.2,0.1,-0.1)^\top$;
	\end{itemize}
	\item Generate exposure variable $Z$ from a Bernoulli distribution with $\Pr(Z=1| L^0)=\text{expit}(\theta_z^\top L^0)$, where $L^0=(A, X^\top)^\top$ and $\theta_z=(0.1,0.2,-0.1,0.1)^\top$, and define $Z \equiv Z_Y \equiv Z_S$;
	
	\item Set $t_{\max}=3$, $\tau=(0,1/4,1/2,1)$, $S^0=1$ and $Y^0=0$;
	\item Generate $L^t$, $S^t$ and $Y^t$ for $t=1,...,t_{\max}$; let $\theta=(1.0,0.5,0.5,0.5)^\top$; 
	
	If $S^{t-1}=1$,
	\begin{itemize}
		\item $L^t$ is generated from a Bernoulli distribution with $\Pr(L^t=1 |S^{t-1}=1, \bar{L}^{t-1}, Z_S,L^0)=\text{expit}(\theta^\top L^0 + L^{t-1} + Z_S)$. The term \(L^{t-1}\) is omitted when \(t=1\);
		\item  $S^t$ is generated from a Bernoulli distribution with $\Pr(S^t=1|S^{t-1}=1, \bar{L}^t, Z_S,L^0)=\text{expit}(-1.1+\theta^\top L^0 + L^{t} + 0.5Z_S)$;
		\item $Y^t$ is generated from a normal distribution $N(0.5+\theta^\top L^0 + L^t + Y^{t-1} + Z_Y, 0.5^2)$ if $S^t=1$.
	\end{itemize}
\end{enumerate}

We are interested in estimating $\Gamma(1,0)-\Gamma(0,0)$ and $\Gamma(1,1)-\Gamma(0,1)$ with different choices of weighting schemes, and their true values are
provided in Table \ref{tab:trueValue_SE}. Table \ref{tab:re_simul_SE} summarizes the simulation results based on sample sizes of $n=500,2000$ with $500$ independent replications. 

\begin{table}[!htb]
	\centering
	\caption{The true values of $\Gamma(1,0)-\Gamma(0,0)$ and $\Gamma(1,1)-\Gamma(0,1)$ under different weighting schemes}
	\label{tab:trueValue_SE}
	\begin{tabular}{lcc}
		\toprule
		Weight  & $\Gamma(1,0)-\Gamma(0,0)(z_S=0)$ & $\Gamma(1,1)-\Gamma(0,1)(z_S=1)$\\
		\midrule
			Exit time & 1.03 & 1.44\\
			Average & 0.51 & 0.72\\
			Cumulative & 1.76 & 2.58\\
			AUC-based & 0.44 & 0.68\\
		\bottomrule
	\end{tabular}
\end{table}

\begin{table}[!htb]
	\centering
	\caption{Bias ($\times 10^3$) and its Monte Carlo standard error ($\times 10^3$, shown in parentheses) for estimating $\Gamma(1,0)-\Gamma(0,0)$ and $\Gamma(1,1)-\Gamma(0,1)$ under different weighting schemes; results are based on sample sizes of $n=500, 2000$ with $500$ independent replications.}
	\label{tab:re_simul_SE}
	\begin{tabular}{llcc}
		\toprule
		Sample size & Weight & $\Gamma(1,0)-\Gamma(0,0)(z_S=0)$ & $\Gamma(1,1)-\Gamma(0,1)(z_S=1)$\\
		\midrule
		\multirow{4}{*}{$n=500$} 
		& Exit time & 8.47 (3.91) & 6.42 (4.63)\\
		& Average & 4.31 (1.95) & 3.36 (2.27)\\
		& Cumulative & 14.71 (7.53) & 11.00 (8.80)\\
		& AUC-based & 3.49 (2.18)  & 2.82 (2.54)\\
		\midrule
		\multirow{4}{*}{$n=2000$} 
		& Exit time & 0.11 (1.98) & $-$2.03 (2.28)\\
		& Average & $-$0.07 (0.99) & $-$1.19 (1.13)\\
		& Cumulative & $-$0.20 (3.73)   & $-$4.73 (4.40)\\
		& AUC-based & $-$0.21 (1.06)   & $-$1.46 (1.27)\\
		\bottomrule
	\end{tabular}
\end{table}

\clearpage
\putbib
\end{bibunit}

\end{document}